\documentclass[11pt,twoside]{article}%{proc}
\usepackage{amssymb,latexsym,graphicx,hyperref,epstopdf,color}%,lscape,}%,subfigure}
\usepackage{epstopdf,caption,subcaption}
\usepackage{chngcntr}
\counterwithin{figure}{section}
\counterwithin{table}{section}
\newcommand{\OPT}{\mbox{\sc OPT}}

\def\mcA{\mathcal{A}}
\def\mcB{\mathcal{B}}
\def\mcC{\mathcal{C}}
\def\mcD{\mathcal{D}}
\def\mcF{\mathcal{F}}
\def\mcJ{\mathcal{J}}
\def\mcO{\mathcal{O}}

\def\mcX{\mathcal{X}}
\def\mcY{\mathcal{Y}}
\normalsize
\newtheorem{theorem}{Theorem}[section]
\newtheorem{lemma}[theorem]{Lemma}

\newtheorem{remark}[theorem]{Remark}
\newenvironment{proof}{{\sc Proof. }}{\hfill$\Box$\vspace{0.2in}}
\title{Approximation algorithms for the three-machine proportionate mixed shop scheduling}
\author{Longcheng~Liu\thanks{School of Mathematical Sciences, Xiamen University. Xiamen, China.
	Email: {\tt longchengliu@xmu.edu.cn}}
	\thanks{Department of Computing Science, University of Alberta. Edmonton, AB T6G 2E8, Canada.
	Emails: {\tt \{rgoebel, guohui\}@ualberta.ca}}
\and
	Yong~Chen\thanks{Department of Mathematics, Hangzhou Dianzi University. Hangzhou, China.
	Emails: {\tt \{chenyong, anzhang\} @hdu.edu.cn}}
	$^{\dagger}$
\and
	Jianming~Dong\thanks{Department of Mathematics, Zhejiang Sci-Tech University. Hangzhou, China.
	Email: {\tt djm226@163.com}}
	$^{\dagger}$
\and
	Randy~Goebel$^{\dagger}$
\and
	Guohui~Lin$^{\dagger}$\thanks{Correspondence author. ORCID: 0000-0003-4283-3396.}
\and
	Yue~Luo$^*$
\and
	Guanqun~Ni\thanks{College of Management, Fujian Agriculture and Forestry University. Fuzhou, China.
	Email: {\tt guanqunni@163.com}}
	$^{\dagger}$
\and
	Bing~Su\thanks{School of Economics and Management, Xi'an Technological University. Xi'an, China.
	Email: {\tt subing684 @sohu.com}}
\and
	An~Zhang$^*$$^{\dagger}$}
\date{\today}
\begin{document}
\maketitle
\begin{abstract}
%==============================================================================
A mixed shop is a manufacturing infrastructure designed to process a mixture of a set of flow-shop jobs and a set of open-shop jobs.
Mixed shops are in general much more complex to schedule than flow-shops and open-shops, and have been studied since the 1980's.
We consider the three machine proportionate mixed shop problem denoted as $M3 \mid prpt \mid C_{\max}$,
in which each job has equal processing times on all three machines.
Koulamas and Kyparisis [{\it European Journal of Operational Research}, 243:70--74,2015] showed that
the problem is solvable in polynomial time in some very special cases;
for the non-solvable case, they proposed a $5/3$-approximation algorithm.
In this paper, we present an improved $4/3$-approximation algorithm and show that this ratio of $4/3$ is asymptotically tight;
when the largest job is a flow-shop job, we present a fully polynomial-time approximation scheme (FPTAS).
On the negative side, while the $F3 \mid prpt \mid C_{\max}$ problem is polynomial-time solvable,
we show an interesting hardness result that adding one open-shop job to the job set makes the problem NP-hard
if this open-shop job is larger than any flow-shop job.
%(denote this problem as $M3 \mid prpt, (n-1, 1), p_1 < q_n \mid C_{\max}$).
We are able to design an FPTAS for this special case too. %the problem $M3 \mid prpt, (n-1, 1), p_1 < q_n \mid C_{\max}$.

\end{abstract}
%==============================================================================

\section{Introduction}
\label{sec1}
%==================================================================================================
We study the following three-machine proportionate mixed shop problem, denoted as $M3 \mid prpt \mid C_{\max}$ in the three-field notation \cite{GLL79}.
Given three machines $M_1, M_2, M_3$ and a set $\mcJ = \mcF \cup \mcO$ of jobs,
where $\mcF = \{J_1, J_2, \ldots, J_\ell\}$ and $\mcO = \{J_{\ell+1}, J_{\ell+2}, \ldots, J_n\}$,
each job $J_i \in \mcF$ needs to be processed non-preemptively through $M_1, M_2, M_3$ sequentially with a processing time $p_i$ on each machine
and each job $J_i \in \mcO$ needs to be processed non-preemptively on $M_1, M_2, M_3$ in any machine order, with a processing time $q_i$ on each machine.
The scheduling constraint is usual in that at every time point a job can be processed by at most one machine and a machine can process at most one job.
The objective is to minimize the maximum job completion time, {\it i.e.}, the makespan.

The jobs of $\mcF$ are referred to as {\em flow-shop jobs} and the jobs of $\mcO$ are called {\em open-shop jobs}.
The mixed shop problem is to process such a mixture of a set of flow-shop jobs and a set of open-shop jobs.
We assume without loss of generality that $p_1 \ge p_2 \ge \ldots \ge p_{\ell}$ and $q_{\ell + 1} \ge q_{\ell + 2} \ge \ldots \ge q_n$.

Mixed shops have many real-life applications and have been studied since the 1980's.
The scheduling of medical tests in an outpatient health care facility and the scheduling of classes/exams in an academic institution are two typical examples,
where the patients (students, respectively) must complete a number of medical tests (academic activities, respectively);
some of these activities must be done in a specified sequential order while the others can be finished in any order;
and the time-spans for all these activities should not overlap with each other.
The {\em proportionate} shops were also introduced in the 1980's \cite{OW85} and
they are one of the most specialized shops with respect to the job processing times which have received many studies \cite{PSK13}.

Masuda et al. \cite{MIN85} and Strusevich \cite{Str91} considered the two-machine mixed shop problem to minimize the makespan, {\it i.e.}, $M2 \mid \mid C_{\max}$;
they both showed that the problem is polynomial-time solvable.
Shakhlevich and Sotskov \cite{SS94} studied mixed shops for processing two jobs with an arbitrary regular objective function.
Brucker \cite{Bru95} surveyed the known results on the mixed shop problems either with two machines or for processing two jobs.
Shakhlevich et al. \cite{SSW99} studied the mixed shop problems with more than two machines for processing more than two jobs, with or without preemption.
Shakhlevich et al. \cite{SSW00} reviewed the complexity results on the mixed shop problems with three or more machines
for processing a constant number of jobs. 

When $\mcO = \emptyset$, the $M3 \mid prpt \mid C_{\max}$ problem reduces to the $F3 \mid prpt \mid C_{\max}$ problem,
which is solvable in polynomial time \cite{CT81}.
When $\mcF = \emptyset$, the problem reduces to the $O3 \mid prpt \mid C_{\max}$ problem, which is ordinary NP-hard \cite{LB87}.
It follows that the $M3 \mid prpt \mid C_{\max}$ problem is at least ordinary NP-hard.
Recently, Koulamas and Kyparisis \cite{KK15} showed that for some very special cases, the $M3 \mid prpt \mid C_{\max}$ problem is solvable in polynomial time;
for the non-solvable case, they showed an absolute performance bound of $2 \max\{p_1, q_{\ell+1}\}$ and presented a $5/3$-approximation algorithm.

In this paper, we design an improved $4/3$-approximation algorithm for (the non-solvable case of) the $M3 \mid prpt \mid C_{\max}$ problem,
and show that the performance ratio of $4/3$ is asymptotically tight.
When the largest job is a flow-shop job, that is $p_1 \ge q_{\ell + 1}$, we present a {\em fully polynomial-time approximation scheme} (FPTAS).
On the negative side, while the $F3 \mid prpt \mid C_{\max}$ problem is polynomial-time solvable,
we show an interesting hardness result that adding one single open-shop job to the job set makes the problem NP-hard
if this open-shop job is larger than any flow-shop job (that is, $\mcF = \{J_1, J_2, \ldots, J_{n-1}\}$ and $\mcO = \{J_n\}$, and $q_n > p_1$).
We construct the reduction from the well-known {\sc Partition} problem \cite{GJ79}.
Denote the special case in which $|\mcF| = n-1$ and $|\mcO| = 1$ as $M3 \mid prpt, (n-1, 1) \mid C_{\max}$.
We propose an FPTAS for this special case $M3 \mid prpt, (n-1, 1) \mid C_{\max}$.

The rest of the paper is organized as follows.
In Section 2, we introduce some notation and present a lower bound on the optimal makespan $C^*_{\max}$.
We present in Section 3 the FPTAS for the $M3 \mid prpt \mid C_{\max}$ problem when $p_1 \ge q_{\ell + 1}$.
The $4/3$-approximation algorithm for the case where $p_1 < q_{\ell + 1}$ is presented in Section 4,
and the performance ratio of $4/3$ is shown to be asymptotically tight.
We show in Section 5 that, when the open-shop job $J_n$ is the only largest, the special case $M3 \mid prpt, (n-1, 1) \mid C_{\max}$ is NP-hard,
through a reduction from the {\sc Partition} problem.
Section 6 contains an FPTAS for the special case $M3 \mid prpt, (n-1, 1) \mid C_{\max}$.
We conclude the paper with some remarks in Section 7.

\section{Preliminaries} \label{sec15}
%==================================================================================================
For any subset of jobs $\mcX \subseteq \mcF$, the \textit{total processing time} of the jobs of $\mcX$ on one machine is denoted as
\[
P(\mcX) = \sum_{J_i \in \mcX} p_i.
\]
For any subset of jobs $\mcY \subseteq \mcO$, the \textit{total processing time} of the jobs of $\mcY$ on one machine is denoted as
\[
Q(\mcY) = \sum_{J_i \in \mcY} q_i.
\]
The set minus operation $\mcJ \setminus \{J\}$ for a single job $J \in \mcJ$ is abbreviated as $\mcJ \setminus J$ throughout the paper.

In a schedule $\pi$, we use $S_j^i$ and $C_j^i$ to denote the start time and the finish time of the job $J_j$ on the machine $M_i$, respectively, 
for $i= 1, 2, 3$ and $j=1, 2, \ldots, n$.

Given that the {\em load} ({\it i.e.}, the total job processing time) of each machine is $P(\mcF) + Q(\mcO)$,
the job $J_{\ell +1}$ has to be processed by all three machines,
and one needs to process all the flow-shop jobs of $\mcF$,
the following lower bound on the optimum $C_{\max}^*$ is established \cite{CT81,KK15}:
\begin{equation}
\label{eq1}
C_{\max}^* \ge \max \{P(\mcF) + Q(\mcO), \ 3q_{\ell + 1}, \ 2 p_1 + P(\mcF)\}.
\end{equation}

\section{An FPTAS for the case where $p_1 \ge q_{\ell + 1}$} \label{sec2}
%==================================================================================================
In this section, we design an approximation algorithm $A(\epsilon)$ for the $M3 \mid prpt \mid C_{\max}$ problem when $p_1 \ge q_{\ell + 1}$,
for any given $\epsilon > 0$.
The algorithm $A(\epsilon)$ produces a schedule $\pi$ with its makespan $C_{\max}^{\pi} < (1 + \epsilon) C_{\max}^*$,
and its running time polynomial in both $n$ and ${1}/{\epsilon}$.

Consider a bipartition $\{ \mcA, \mcB \}$ of the job set $\mcO = \{J_{\ell+1}, J_{\ell+2}, \ldots, J_n\}$,
{\it i.e.}, $\mcA \cup \mcB = \mcO$ and $\mcA \cap \mcB = \emptyset$.
Throughout the paper, a part of the bipartition is allowed to be empty.
The following six-step \textit{procedure} {\sc Proc}$(\mcA, \mcB, \mcF)$ produces a schedule $\pi$:

\begin{enumerate}
\item the jobs of $\mcF$ are processed in the same {\em longest processing time} (LPT) order on all three machines, 
	and every job is processed first on $M_1$, then on $M_2$, lastly on $M_3$;

\item the jobs of $\mcA$ are processed in the same LPT order on all three machines,
	and every one is processed first on $M_2$, then on $M_3$, lastly on $M_1$;

\item the jobs of $\mcB$ are processed in the same LPT order on all three machines,
	and every one is processed first on $M_3$, then on $M_1$, lastly on $M_2$;
	and

\item the machine $M_1$ processes (the jobs of) $\mcF$ first, then $\mcB$, lastly $\mcA$, denoted as $\langle\mcF, \mcB, \mcA\rangle$;

\item the machine $M_2$ processes $\mcA$ first, then $\mcF$, lastly $\mcB$, denoted as $\langle\mcA, \mcF, \mcB\rangle$;

\item the machine $M_3$ processes $\mcB$ first, then $\mcA$, lastly $\mcF$, denoted as $\langle\mcB, \mcA, \mcF\rangle$.
\end{enumerate}
{\sc Proc}$(\mcA, \mcB, \mcF)$ runs in $O(n \log n)$ time to produce the schedule $\pi$, of which an illustration is shown in Figure~\ref{fig31}.

\begin{figure}[ht]
\centering
  \setlength{\unitlength}{0.9bp}%
  \begin{picture}(201.85, 99.05)(0,0)
  \put(20,5){\includegraphics[scale=0.9]{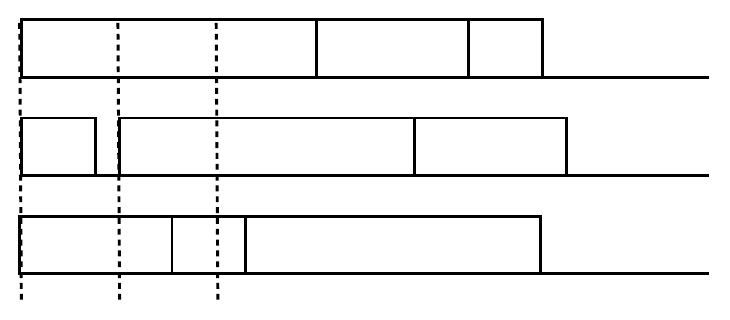}}
  \put(0.0,78.596){\fontsize{14.23}{17.07}\selectfont $M_1$}
  \put(0.0,49.73){\fontsize{14.23}{17.07}\selectfont $M_2$}
  \put(0.0,20.73){\fontsize{14.23}{17.07}\selectfont $M_3$}
  \put(23.0,0.0){\fontsize{14.23}{17.07}\selectfont $0$}
  \put(50.0,0.0){\fontsize{14.23}{17.07}\selectfont $p_1$}
  \put(70.97,0.0){\fontsize{14.23}{17.07}\selectfont $2p_1$}  
  \put(60.23,78.596){\fontsize{14.23}{17.07}\selectfont $\mcF$}
  \put(90.23,49.73){\fontsize{14.23}{17.07}\selectfont $\mcF$}
  \put(130.23,20.73){\fontsize{14.23}{17.07}\selectfont $\mcF$}
  \put(160.23,78.596){\fontsize{14.23}{17.07}\selectfont $\mcA$}
  \put(30.23,49.73){\fontsize{14.23}{17.07}\selectfont $\mcA$}
  \put(73.23,20.73){\fontsize{14.23}{17.07}\selectfont $\mcA$}
  \put(130.23,78.596){\fontsize{14.23}{17.07}\selectfont $\mcB$}
  \put(155.23,49.73){\fontsize{14.23}{17.07}\selectfont $\mcB$}
  \put(38.23,20.73){\fontsize{14.23}{17.07}\selectfont $\mcB$}
  \end{picture}%
\caption{An illustration of the schedule $\pi$ produced by {\sc Proc}$(\mcA, \mcB, \mcF)$,
	where $\{\mcA, \mcB\}$ is a bipartition of the set $\mcO$ and
	the jobs of each of $\mcA, \mcB, \mcF$ are processed in the same LPT order on all three machines.\label{fig31}}
\end{figure}
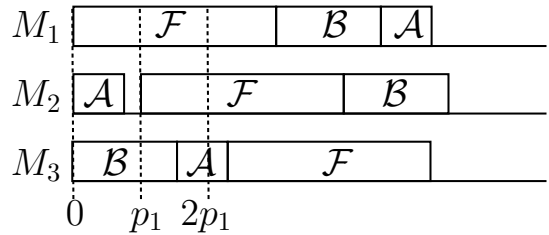

\begin{lemma}
\label{lemma31}
If both $Q(\mcA) \le p_1$ and $Q(\mcB) \le p_1$, then the schedule $\pi$ produced by {\sc Proc}$(\mcA, \mcB, \mcF)$ is optimal, 
with its makespan $C_{\max}^{\pi} = C_{\max}^* = 2p_1 + P(\mcF)$.
\end{lemma}
\begin{proof}
The schedule depicted in Figure~\ref{fig32} is feasible since we have the following inequalities:
$Q(\mcA) \le p_1$, $Q(\mcB) \le p_1$.
The makespan is achieved on $M_3$ and $C_{\max}^{\pi} = 2p_1 + P(\mcF)$,
which meets the lower bound in Eq.~(\ref{eq1}).
\end{proof}

\begin{figure}[ht]
\centering
  \setlength{\unitlength}{0.9bp}%
  \begin{picture}(201.85, 99.05)(0,0)
  \put(20,4){\includegraphics[scale=0.9]{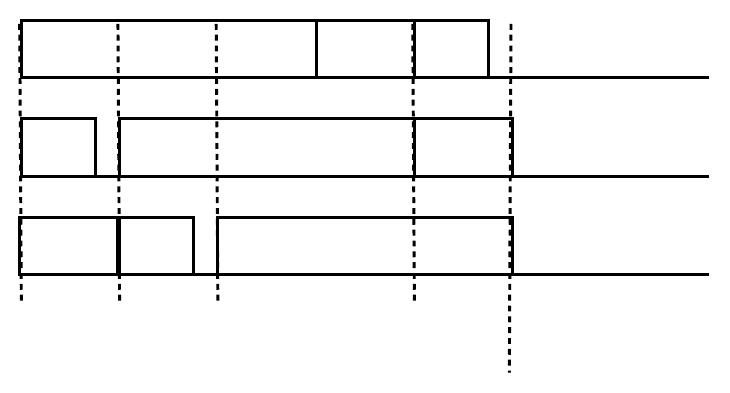}}
  \put(0.0,98.596){\fontsize{14.23}{17.07}\selectfont $M_1$}
  \put(0.0,69.73){\fontsize{14.23}{17.07}\selectfont $M_2$}
  \put(0.0,40.73){\fontsize{14.23}{17.07}\selectfont $M_3$}
  \put(23.0,20.0){\fontsize{10.23}{17.07}\selectfont $0$}
  \put(50.0,20.0){\fontsize{10.23}{17.07}\selectfont $p_1$}
  \put(76.97,20.0){\fontsize{10.23}{17.07}\selectfont $2p_1$}
  \put(110.97,20.0){\fontsize{10.23}{17.07}\selectfont $p_1 + P(\mcF)$}
  \put(155.97,0.0){\fontsize{10.23}{17.07}\selectfont $2p_1 + P(\mcF)$}  
  \put(60.23,98.596){\fontsize{14.23}{17.07}\selectfont $\mcF$}
  \put(95.23,69.73){\fontsize{14.23}{17.07}\selectfont $\mcF$}
  \put(120.23,40.73){\fontsize{14.23}{17.07}\selectfont $\mcF$}
  \put(145.23,98.596){\fontsize{14.23}{17.07}\selectfont $\mcA$}
  \put(30.23,69.73){\fontsize{14.23}{17.07}\selectfont $\mcA$}
  \put(58.23,40.73){\fontsize{14.23}{17.07}\selectfont $\mcA$}
  \put(120.23,98.596){\fontsize{14.23}{17.07}\selectfont $\mcB$}
  \put(148.23,69.73){\fontsize{14.23}{17.07}\selectfont $\mcB$}
  \put(35.23,40.73){\fontsize{14.23}{17.07}\selectfont $\mcB$}
  \end{picture}%
\caption{An illustration of the schedule $\pi$ produced by {\sc Proc}$(\mcA, \mcB, \mcF)$ when both $Q(\mcA) \le p_1$ and $Q(\mcB) \le p_1$.\label{fig32}}
\end{figure}

\begin{lemma}
\label{lemma32}
If both $Q(\mcA) \ge p_1$ and $Q(\mcB) \ge p_1$, then the schedule $\pi$ produced by {\sc Proc}$(\mcA, \mcB, \mcF)$ is optimal, 
with its makespan $C_{\max}^{\pi} = C_{\max}^* = P(\mcF) + Q(\mcO)$.
\end{lemma}
\begin{proof}
In this case, in the schedule $\pi$ produced by {\sc Proc}$(\mcA, \mcB, \mcF)$,
the machine $M_1$ does not idle since $p_1 \ge q_{\ell + 1}$ and the jobs of $\mcA$, $\mcB$ and $\mcF$ are processed in the LPT order.
The machine $M_2$ does not idle either since $Q(\mcA) \ge p_1$ and the jobs of $\mcA$, $\mcB$ and $\mcF$ are processed in the LPT order.
The machine $M_3$ does not idle since $Q(\mcB) \ge p_1 \ge q_{\ell + 1}$, $Q(\mcA) + Q(\mcB) \ge 2p_1$,
and the jobs of $\mcA$, $\mcB$ and $\mcF$ are processed in the LPT order.
Therefore, the makespan of this schedule is $C_{\max}^{\pi} = P(\mcF) + Q(\mcO)$,
which meets the lower bound in Eq.~(\ref{eq1}) (see for an illustration in Figure~\ref{fig33}).
This finishes the proof of the lemma.
\end{proof}

\begin{figure}[ht]
\centering
  \setlength{\unitlength}{0.9bp}%
  \begin{picture}(201.85, 89.05)(0,0)
  \put(20,5){\includegraphics[scale=0.9]{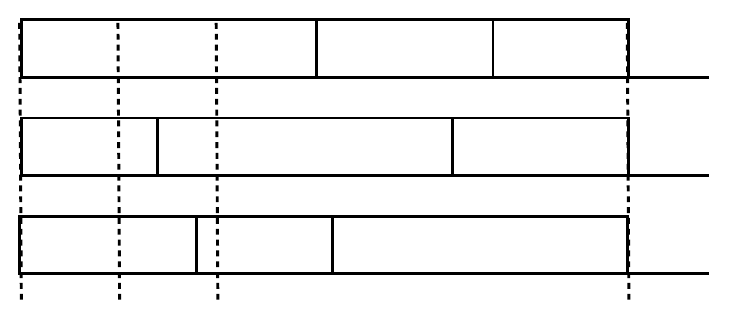}}
  \put(0.0,78.596){\fontsize{14.23}{17.07}\selectfont $M_1$}
  \put(0.0,49.73){\fontsize{14.23}{17.07}\selectfont $M_2$}
  \put(0.0,20.73){\fontsize{14.23}{17.07}\selectfont $M_3$}
  \put(23.0,0.0){\fontsize{14.23}{17.07}\selectfont $0$}
  \put(50.0,0.0){\fontsize{14.23}{17.07}\selectfont $p_1$}
  \put(70.97,0.0){\fontsize{14.23}{17.07}\selectfont $2p_1$}
  \put(160.97,0.0){\fontsize{14.23}{17.07}\selectfont $P(\mcF) + Q(\mcO)$}  
  \put(60.23,78.596){\fontsize{14.23}{17.07}\selectfont $\mcF$}
  \put(105.23,49.73){\fontsize{14.23}{17.07}\selectfont $\mcF$}
  \put(140.23,20.73){\fontsize{14.23}{17.07}\selectfont $\mcF$}
  \put(175.23,78.596){\fontsize{14.23}{17.07}\selectfont $\mcA$}
  \put(40.23,49.73){\fontsize{14.23}{17.07}\selectfont $\mcA$}
  \put(88.23,20.73){\fontsize{14.23}{17.07}\selectfont $\mcA$}
  \put(135.23,78.596){\fontsize{14.23}{17.07}\selectfont $\mcB$}
  \put(168.23,49.73){\fontsize{14.23}{17.07}\selectfont $\mcB$}
  \put(55.23,20.73){\fontsize{14.23}{17.07}\selectfont $\mcB$}
  \end{picture}%
\caption{An illustration of the schedule $\pi$ produced by {\sc Proc}$(\mcA, \mcB, \mcF)$ when both $Q(\mcA) \ge p_1$ and $Q(\mcB) \ge p_1$.\label{fig33}}
\end{figure}

Now we are ready to present the approximation algorithm $A(\epsilon)$, for any $\epsilon >0$.

In the first step, 
we check whether $Q(\mcO) \le p_1$ or not.
If $Q(\mcO) \le p_1$, then we run {\sc Proc}$(\mcO, \emptyset, \mcF)$ to construct a schedule $\pi$ and terminate the algorithm.
The schedule $\pi$ is optimal by Lemma \ref{lemma31}.

In the second step, 
the algorithm $A(\epsilon)$ constructs an instance of the {\sc Knapsack} problem \cite{GJ79},
in which there is an item corresponding to the job $J_i \in \mcO$, also denoted as $J_i$.
The item $J_i$ has a profit $q_i$ and a size $q_i$.
The capacity of the knapsack is $p_1$.
The {\sc Min-Knapsack} problem is to find a subset of items of minimum profit that \textit{cannot} be packed into the knapsack,
and it admits an FPTAS \cite{KPP04}. 
The algorithm $A(\epsilon)$ runs a $(1 + \epsilon)$-approximation algorithm for the {\sc Min-Knapsack} problem to obtain a job subset $\mcA$.
It then runs {\sc Proc}$(\mcA, \mcO \setminus \mcA, \mcF)$ to construct a schedule, denoted as $\pi^1$.

The {\sc Max-Knapsack} problem is to find a subset of items of maximum profit that can be packed into the knapsack, and it admits an FPTAS, too \cite{KP04}. 
In the third step, 
the algorithm $A(\epsilon)$ runs a $(1 - \epsilon)$-approximation algorithm for the {\sc Max-Knapsack} problem to obtain a job subset $\mcB$.
Then it runs {\sc Proc}$(\mcO \setminus \mcB, \mcB, \mcF)$ to construct a schedule, denoted as $\pi^2$.

The algorithm $A(\epsilon)$ outputs the schedule with a smaller makespan between $\pi^1$ and $\pi^2$. 
A high-level description of the algorithm $A(\epsilon)$ is provided in Figure~\ref{fig34}.

\begin{figure}[ht]
\begin{center}
\framebox{
\begin{minipage}{5.0in}
	{\sc Algorithm}  $A(\epsilon)$:
	\begin{enumerate}%[{Step} 1.]
		\item If $Q(\mcO) \le p_1$, 
		      then run {\sc Proc}($\mcO, \emptyset, \mcF$) to produce a schedule $\pi$;
		      output the schedule $\pi$.
		\item Construct an instance of {\sc Knapsack}, 
		      where an item $J_i$ corresponds to the job $J_i \in \mcO$;
		      $J_i$ has a profit $q_i$ and a size $q_i$;
		      the capacity of the knapsack is $p_1$. 
			\begin{description}
			\parskip=0pt
		    \item[2.1.] Run a $(1+\epsilon)$-approximation for {\sc Min-Knapsack} to obtain a job subset $\mcA$. 
		    \item[2.2.]	Run {\sc Proc}($\mcA, \mcO \setminus  \mcA, \mcF$) to construct a schedule $\pi^1$.
			\end{description}
        \item 
			\begin{description}
			\parskip=0pt
		    \item[3.1.] Run a $(1-\epsilon)$-approximation for {\sc Max-Knapsack} to obtain a job subset $\mcB$.
		    \item[3.2.]	Run {\sc Proc}($\mcO \setminus  \mcB, \mcB, \mcF$) to construct a schedule $\pi^2$.
			\end{description}
        \item Output the schedule with a smaller makespan between $\pi^1$ and $\pi^2$.		
	\end{enumerate}
\end{minipage}}
\end{center}
\caption{A high-level description of the algorithm $A(\epsilon)$.\label{fig34}}
\end{figure}

In the following performance analysis, we assume without of loss of generality that $Q(\mcO) > p_1$.
We have the following (in-)equalities inside the algorithm $A(\epsilon)$:

\begin{eqnarray}
\OPT^1 &= &\min \{ Q(\mcX) \mid \mcX \subseteq \mcO,~ Q(\mcX) > p_1 \};\label{eq2}\\
p_1    &< &Q(\mcA) \ \ \le \ \ (1+\epsilon) \OPT^1;\label{eq3}\\
\OPT^2 &= &\max \{ Q(\mcY) \mid \mcY \subseteq \mcO,~ Q(\mcY) \le p_1 \};\label{eq4}\\
p_1    &\ge &Q(\mcB) \ \ \ge \ \ (1-\epsilon) \OPT^2,\label{eq5}
\end{eqnarray}
where $\OPT^1$ ($\OPT^2$, respectively) is the optimum to the constructed {\sc Min-Knapsack} ({\sc Max-Knapsack}, respectively) problem.

\begin{lemma}
\label{lemma33}
In the algorithm $A(\epsilon)$,
if $Q(\mcO \setminus \mcA) \le p_1 - \epsilon \OPT^1$,
then for any  bipartition $\{\mcX, \mcY\}$ of the job set $\mcO$, $Q(\mcX) > p_1$ implies $Q(\mcY) \le p_1$.
\end{lemma}
\begin{proof}
Note that the job subset $\mcA$ is computed in Step 2.1 of the algorithm $A(\epsilon)$, and it satisfies Eq.~(\ref{eq3}).
By the definition of $\OPT^1$ in Eq.~(\ref{eq2}) and using Eq.~(\ref{eq3}), we have $Q(\mcX) \ge \OPT^1 \ge Q(\mcA) - \epsilon \OPT^1$.
Furthermore, from the fact that $Q(\mcO) = Q(\mcX) + Q(\mcY) = Q(\mcA) + Q(\mcO \setminus \mcA)$ 
and the assumption that $Q(\mcO \setminus \mcA) \le p_1 - \epsilon \OPT^1$,
we have
\begin{eqnarray*}
Q(\mcY) & =   & Q(\mcA) + Q(\mcO \setminus \mcA) - Q(\mcX) \\
        & \le & Q(\mcA) + Q(\mcO \setminus \mcA) - (Q(\mcA) - \epsilon \OPT^1) \\
        & =   & Q(\mcO \setminus \mcA) + \epsilon \OPT^1 \\
        & \le & p_1 - \epsilon \OPT^1 + \epsilon \OPT^1 \\
        & =   & p_1.
\end{eqnarray*}
This finishes the proof of the lemma.
\end{proof}

\begin{lemma}
\label{lemma34}
In the algorithm $A(\epsilon)$,
if $Q(\mcO \setminus \mcA) \le p_1 - \epsilon \OPT^1$, then $C_{\max}^* \ge P(\mcF) + Q(\mcO) + p_1 - \OPT^2.$
\end{lemma}
\begin{proof}
Consider an arbitrary optimal schedule $\pi^*$ that achieves the makespan $C_{\max}^*$.
Note that the flow-shop job $J_1$ is first processed on the machine $M_1$, 
then on machine $M_2$,
and last on machine $M_3$.

%In the schedule $\pi^*$, let $S_i$ and $C_i$ be the start time and the finish time of the job $J_1$ on the machine $M_i$, respectively,
%for $i = 1, 2, 3$.
On the machine $M_2$, let $\mcJ^1 = \mcO^1 \cup \mcF^1$ denote the subset of jobs processed before $J_1$,
and $\mcJ^2 = \mcO^2 \cup \mcF^2$ denote the subset of jobs processed after $J_1$,
where $\{\mcO^1, \mcO^2\}$ is a bipartition of the job set $\mcO$ and $\{\mcF^1, \mcF^2\}$ is a bipartition of the job set $\mcF \setminus J_1$.
Also, let $\delta_1$ and $\delta_2$ denote the total amount of machine idle time for $M_2$ before processing $J_1$ and after processing $J_1$, respectively
(see Figure~\ref{fig35} for an illustration).

\begin{figure}[ht]
\centering
  \setlength{\unitlength}{0.9bp}%
  \begin{picture}(201.85, 99.05)(0,0)
  \put(20,5){\includegraphics[scale=0.9]{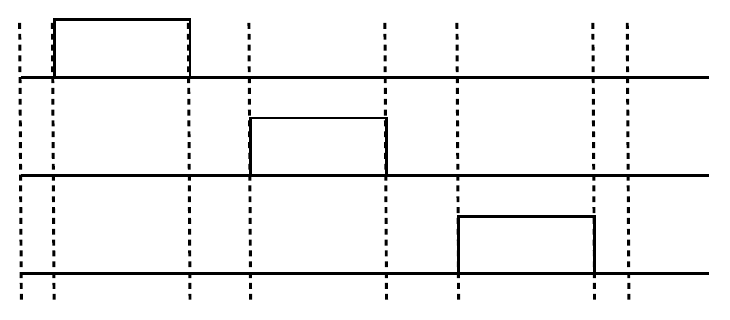}}
  \put(0.0,78.596){\fontsize{14.23}{17.07}\selectfont $M_1$}
  \put(0.0,49.73){\fontsize{14.23}{17.07}\selectfont $M_2$}
  \put(0.0,20.73){\fontsize{14.23}{17.07}\selectfont $M_3$}
  \put(23.0,0.0){\fontsize{10.23}{17.07}\selectfont $0$}
  \put(32.0,0.0){\fontsize{10.23}{17.07}\selectfont $S_1^1$}
  \put(88.0,0.0){\fontsize{10.23}{17.07}\selectfont $S_1^2$}
  \put(150.0,0.0){\fontsize{10.23}{17.07}\selectfont $S_1^3$}
  \put(70.97,0.0){\fontsize{10.23}{17.07}\selectfont $C_1^1$}
  \put(125.97,0.0){\fontsize{10.23}{17.07}\selectfont $C_1^2$}
  \put(185.97,0.0){\fontsize{10.23}{17.07}\selectfont $C_1^3$}
  \put(200.97,0.0){\fontsize{10.23}{17.07}\selectfont $C_{\max}^*$}  
  \put(50.23,78.596){\fontsize{14.23}{17.07}\selectfont $J_1$}
  \put(105.23,49.73){\fontsize{14.23}{17.07}\selectfont $J_1$}
  \put(165.23,20.73){\fontsize{14.23}{17.07}\selectfont $J_1$}
  \put(40.23,49.73){\fontsize{14.23}{17.07}\selectfont $\mcJ^1$}
  \put(80.23,49.73){\fontsize{14.23}{17.07}\selectfont $\delta_1$}
  \put(168.23,49.73){\fontsize{14.23}{17.07}\selectfont $\mcJ^2$}
  \put(140.23,49.73){\fontsize{14.23}{17.07}\selectfont $\delta_2$}
  \end{picture}%
\caption{An illustration of an optimal schedule $\pi^*$,
	in which $\mcJ^1$ and $\mcJ^2$ are the subsets of jobs processed on $M_2$ before $J_1$ and after $J_1$, respectively;
	$\delta_1$ and $\delta_2$ are the total amount of machine idle time for $M_2$ before processing $J_1$ and after processing $J_1$, respectively.\label{fig35}}
\end{figure}
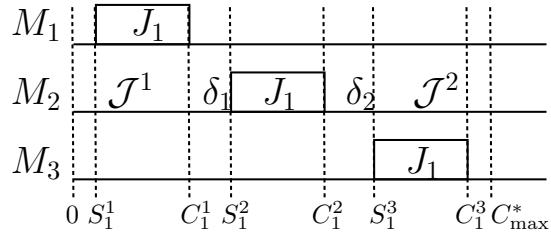

Note that $\mcF = J_1 \cup \mcF^1 \cup \mcF^2$ is the set of flow-shop jobs.
The job $J_1$ and the jobs of $\mcF^1$ should be finished on the machine $M_1$ before time $S_1^2$,
and the job $J_1$ and the jobs of $\mcF^2$ can only be started on the machine $M_3$ after time $C_1^2$.
That is,
\begin{equation}
\label{eq6}
p_1 + P(\mcF^1) \le S_1^2
\end{equation}
and
\begin{equation}
\label{eq7}
p_1 + P(\mcF^2) \le C_{\max}^* - C_1^2.
\end{equation}

If $Q(\mcO^1) \le p_1$, 
then we have $Q(\mcO^1) \le \OPT^2$ by the definition of $\OPT^2$ in Eq.~(\ref{eq4}).
Combining this with Eq.~(\ref{eq6}), 
we achieve that $\delta_1 = S_1^2 - P(\mcF^1) - Q(\mcO^1) \ge p_1- \OPT^2$.

If $Q(\mcO^1) > p_1$, 
then we have $Q(\mcO^2) \le p_1$ by Lemma \ref{lemma33}.
Hence, $Q(\mcO^2) \le \OPT^2$ by the definition of $\OPT^2$ in Eq.~(\ref{eq4}).
Combining this with Eq.~(\ref{eq7}), 
we achieve that $\delta_2 = C_{\max}^* - C_1^2 - P(\mcF^2) - Q(\mcO^2) \ge p_1- \OPT^2$.
 
The last two paragraphs prove that $\delta_1 + \delta_2 \ge p_1- \OPT^2$.
Therefore,
\begin{eqnarray*}
C_{\max}^* 
        & =   & Q(\mcO^1) + P(\mcF^1) + \delta_1 + p_1 + Q(\mcO^2) + P(\mcF^2) + \delta_2 \\
        & =   & P(\mcF) + Q(\mcO) +  \delta_1 + \delta_2 \\
        & \ge & P(\mcF) + Q(\mcO) + p_1- \OPT^2.
\end{eqnarray*}
This finishes the proof of the lemma.
\end{proof}

\begin{lemma}
\label{lemma35}
In the algorithm $A(\epsilon)$, if $Q(\mcO \setminus \mcA) \le p_1 - \epsilon \OPT^1$, then $C_{\max}^{\pi^2} < (1+\epsilon) C_{\max}^*$.
\end{lemma}
\begin{proof}
Denote $\overline{\mcB} = \mcO \setminus \mcB$.
Note that the job set $\mcB$ computed in Step 3.1 of the algorithm $A(\epsilon)$ satisfies $p_1 \ge Q(\mcB) \ge (1 - \epsilon) \OPT^2$, 
and the schedule $\pi^2$ is constructed by {\sc Proc}($\overline{\mcB}, \mcB, \mcF$).
We distinguish the following two cases according to the value of $Q(\overline{\mcB})$.

{Case 1.} $Q(\overline{\mcB}) \le p_1$. 
In this case, 
the schedule $\pi^2$ is optimal by Lemma \ref{lemma31}.

{Case 2.} $Q(\overline{\mcB}) > p_1$.
The schedule $\pi^2$ constructed by {\sc Proc}($\overline{\mcB}, \mcB, \mcF$) has the following properties (see Figure~\ref{fig36} for an illustration):

\begin{figure}[ht]
\centering
  \setlength{\unitlength}{0.89bp}%
  \begin{picture}(251.85, 99.05)(0,0)
  \put(20,5){\includegraphics[scale=0.9]{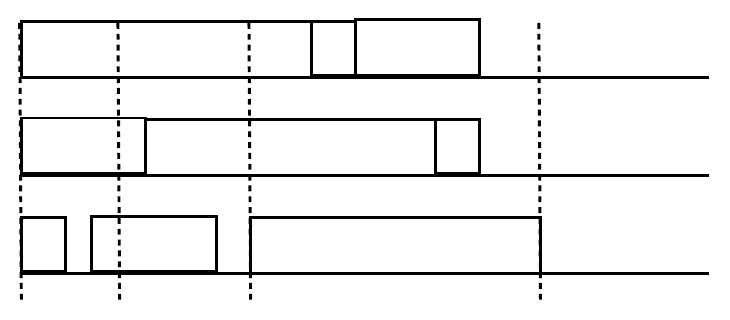}}
  \put(0.0,78.596){\fontsize{14.23}{17.07}\selectfont $M_1$}
  \put(0.0,49.73){\fontsize{14.23}{17.07}\selectfont $M_2$}
  \put(0.0,20.73){\fontsize{14.23}{17.07}\selectfont $M_3$}
  \put(23.0,0.0){\fontsize{14.23}{17.07}\selectfont $0$}
  \put(50.0,0.0){\fontsize{14.23}{17.07}\selectfont $p_1$}
  \put(80.97,0.0){\fontsize{14.23}{17.07}\selectfont $p_1 + Q(\overline{\mcB})$}
  \put(160.97,0.0){\fontsize{14.23}{17.07}\selectfont $P(\mcF) + p_1 + Q(\overline{\mcB})$}  
  \put(60.23,78.596){\fontsize{14.23}{17.07}\selectfont $\mcF$}
  \put(105.23,49.73){\fontsize{14.23}{17.07}\selectfont $\mcF$}
  \put(140.23,20.73){\fontsize{14.23}{17.07}\selectfont $\mcF$}
  \put(110.23,78.596){\fontsize{14.23}{17.07}\selectfont $\mcB$}
  \put(40.23,48.73){\fontsize{14.23}{17.07}\selectfont $\overline{\mcB}$}
  \put(27.23,20.73){\fontsize{14.23}{17.07}\selectfont $\mcB$}
  \put(135.23,77.596){\fontsize{14.23}{17.07}\selectfont $\overline{\mcB}$}
  \put(145.23,49.73){\fontsize{14.23}{17.07}\selectfont $\mcB$}
  \put(58.23,20.73){\fontsize{14.23}{17.07}\selectfont $\overline{\mcB}$}
  \end{picture}%
\caption{An illustration of the schedule $\pi^2$ constructed by {\sc Proc}($\overline{\mcB}, \mcB, \mcF$) in Case 2,
	where $Q(\mcB) \le p_1$ and $Q(\overline{\mcB}) > p_1$.
	The machines $M_1$ and $M_2$ do not idle;
	the machine $M_3$ may idle between processing the job set $\mcB$ and the job set $\overline{\mcB}$
	and may idle between processing the job set $\overline{\mcB}$ and the job set $\mcF$.
	$M_3$ starts processing the job set $\mcF$ at time $p_1 + Q(\overline{\mcB})$.\label{fig36}}
\end{figure}

\begin{enumerate}
\item The jobs are processed consecutively on the machine $M_1$ since $J_1$ is the largest job. 
	The completion time of $M_1$ is thus $C_1^{\pi^2} = Q(\mcO) + P(\mcF)$.

\item The jobs are processed consecutively on the machine $M_2$ due to $Q(\mcB) \le p_1$ and $Q(\overline{\mcB}) > p_1$. 
	The completion time of $M_2$ is thus $C_2^{\pi^2} = Q(\mcO) + P(\mcF)$.

\item The machine $M_3$ starts processing the job set $\mcF$ consecutively at time $p_1 + Q(\overline{\mcB})$ due to $Q(\mcB) \le p_1$.
	The completion time of $M_3$ is $C_3^{\pi^2} =P(\mcF) + p_1 + Q(\overline{\mcB})$.
\end{enumerate}

Note that $C_3^{\pi^2} =P(\mcF) + p_1 + Q(\overline{\mcB}) \ge P(\mcF) + Q(\mcB) + Q(\overline{\mcB}) = Q(\mcO) + P(\mcF)$,
implying $C_{\max}^{\pi^2} = P(\mcF) + p_1 + Q(\overline{\mcB})$.
Combining Eq.~(\ref{eq5}) with Lemma \ref{lemma34}, we have 
\begin{eqnarray*}
C_{\max}^{\pi^2} 
        & =   & P(\mcF) + p_1 + Q(\overline{\mcB}) \\
%        & =   & P(\mcF) + Q(\mcB) + Q(\overline{\mcB}) + p_1 - Q(\mcB)\\
        & =   & P(\mcF) + Q(\mcO)  + p_1 - Q(\mcB)\\
        & \le & P(\mcF) + Q(\mcO) + p_1- (1-\epsilon) \OPT^2\\
%        & =   & P(\mcF) + Q(\mcO) + p_1- \OPT^2 + \epsilon \OPT^2\\
        & \le & C_{\max}^* + \epsilon \OPT^2\\
        & <   & (1 + \epsilon) C_{\max}^*,
\end{eqnarray*}
where the last inequality is due to $\OPT^2 \le p_1 <C_{\max}^*$.
This finishes the proof of the lemma.
\end{proof}

\begin{lemma}
\label{lemma36}
In the algorithm $A(\epsilon)$, if $p_1 - \epsilon \OPT^1 < Q(\mcO \setminus \mcA) < p_1$, then $C_{\max}^{\pi^1} < (1+\epsilon) C_{\max}^*$.
\end{lemma}
\begin{proof}
Denote $\overline{\mcA} = \mcO \setminus \mcA$. 
Note that the job set $\mcA$ computed in Step 2.1 of the algorithm $A(\epsilon)$ satisfies $p_1 < Q(\mcA) \le (1 + \epsilon) \OPT^1$,
and the schedule $\pi^1$ is constructed by {\sc Proc}($\mcA, \overline{\mcA}, \mcF$).

By a similar argument as in Case 2 in the proof of Lemma \ref{lemma35},
replacing the two job sets $\mcB, \overline{\mcB}$ by the two job sets $\overline{\mcA}, \mcA$,
we conclude that the makespan of the schedule $\pi^1$ is achieved on the machine $M_3$, $C_{\max}^{\pi^1} = P(\mcF) + Q(\mcO) + p_1 - Q(\overline{\mcA})$.
Combining Eq.~(\ref{eq1}) with the assumption that $p_1 - \epsilon \OPT^1 < Q(\overline{\mcA})$,
we have
\[
C_{\max}^{\pi^1} %= P(\mcF) + Q(\mcO) + p_1 - Q(\overline{\mcA})
< P(\mcF) + Q(\mcO) + \epsilon \OPT^1 \le C_{\max}^* + \epsilon \OPT^1
< (1 + \epsilon) C_{\max}^*,
\]  
where the last inequality follows from $\OPT^1 \le Q(\mcO) \le C_{\max}^*$.
This finishes the proof of the lemma.
\end{proof}

\begin{theorem}
\label{thm37}
The algorithm $A(\epsilon)$ is an
$O(n \min\{\log n, \log{(1/\epsilon)}\} + 1/{\epsilon}^2 \log{(1/\epsilon)} \min\{n$, $1/{\epsilon} \log{(1/\epsilon)}\})$-time
$(1 + \epsilon)$-approximation for the $M3 \mid prpt \mid C_{\max}$ problem when $p_1 \ge q_{\ell +1}$.
\end{theorem}
\begin{proof}
First of all, the procedure {\sc Proc}($\mcX,  \mcY, \mcF$) on a bipartition $\{\mcX, \mcY\}$ of the job set $\mcO$ takes $O(n \log n)$ time.
Recall that the job set $\mcA$ is computed by a $(1 + \epsilon)$-approximation for the {\sc Min-Knapsack} problem, 
in $O(n \min\{\log n, \log{(1/\epsilon)}\} + 1/{\epsilon}^2 \log{(1/\epsilon)} \min\{n, 1/{\epsilon} \log{(1/\epsilon)}\})$ time;
%which takes a polynomial time in both $n$ and ${1}/{\epsilon}$;
the other job set $\mcB$ is computed by a $(1 - \epsilon)$-approximation for the {\sc Max-Knapsack} problem, 
also in $O(n \min\{\log n, \log{(1/\epsilon)}\} + 1/{\epsilon}^2 \log{(1/\epsilon)} \min\{n, 1/{\epsilon} \log{(1/\epsilon)}\})$ time.
%which also takes a polynomial time in both $n$ and ${1}/{\epsilon}$.
The total running time of the algorithm $A(\epsilon)$ is thus polynomial in both $n$ and ${1}/{\epsilon}$ too.

When $Q(\mcO) \le p_1$, or the job set $\mcO \setminus \mcA$ computed in Step 2.1 of the algorithm $A_1(\epsilon)$ has total processing time not less than $p_1$,
the schedule constructed in the algorithm $A(\epsilon)$ is optimal by Lemmas \ref{lemma31} and \ref{lemma32}.  
When $Q(\mcO \setminus \mcA) < p_1$, the smaller makespan between the two schedules $\pi^1$ and $\pi^2$ constructed by 
the algorithm $A(\epsilon)$ is less than $(1 + \epsilon)$ of the optimum by Lemmas \ref{lemma35} and \ref{lemma36}.
Therefore, the algorithm $A(\epsilon)$ has a worst-case performance ratio of $(1 + \epsilon)$.   
This finishes the proof of the theorem.
\end{proof}

\section{A ${4}/{3}$-approximation for the case where $p_1 < q_{\ell + 1}$}
\label{sec3}
%==================================================================================================
In this section, 
we present a $4/3$-approximation algorithm for the $M3 \mid prpt \mid C_{\max}$ problem when $p_1 < q_{\ell + 1}$,
and we show that this ratio of $4/3$ is asymptotically tight.

\begin{theorem}
\label{thm41}
When $p_1 < q_{\ell + 1}$, the $M3 \mid prpt \mid C_{\max}$ problem admits an $O(n \log n)$-time $4/3$-approximation algorithm.
\end{theorem}
\begin{proof}
Consider first the case where there are at least two open-shop jobs.
Construct a permutation schedule $\pi$ in which the job processing order for $M_1$ is $\langle J_{\ell+3}, \ldots, J_n, \mcF, J_{\ell+1}, J_{\ell+2}\rangle$,
where the jobs of $\mcF$ are processed in the LPT order;
the job processing order for $M_2$ is $\langle J_{\ell+2}, J_{\ell+3}, \ldots, J_n, \mcF, J_{\ell+1}\rangle$;
the job processing order for $M_3$ is $\langle J_{\ell+1}, J_{\ell+2}, J_{\ell+3}, \ldots, J_n, \mcF\rangle$.
See Figure~\ref{fig41} for an illustration, where the start time for $J_{\ell+3}$ on $M_2$ is $q_{\ell+1}$,
and the start time for $J_{\ell+3}$ on $M_3$ is $2 q_{\ell+1}$.
One can check that the schedule $\pi$ is feasible when $p_1 <q_{\ell + 1}$, and it can be constructed in $O(n \log n)$ time.

\begin{figure}[ht]
\centering
  \setlength{\unitlength}{0.9bp}%
  \begin{picture}(301.85, 99.05)(0,0)
  \put(20,5){\includegraphics[scale=0.9]{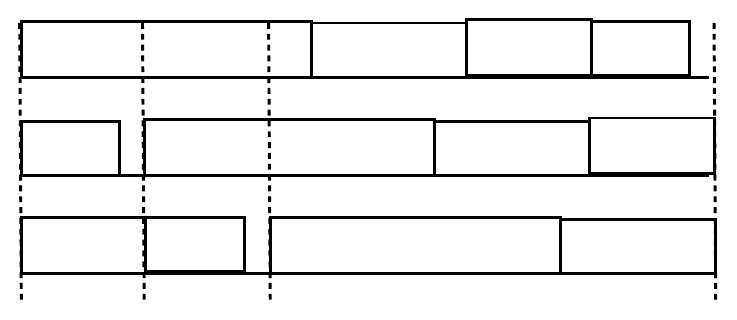}}
  \put(0.0,78.596){\fontsize{14.23}{17.07}\selectfont $M_1$}
  \put(0.0,49.73){\fontsize{14.23}{17.07}\selectfont $M_2$}
  \put(0.0,20.73){\fontsize{14.23}{17.07}\selectfont $M_3$}
  \put(23.0,0.0){\fontsize{14.23}{17.07}\selectfont $0$}
  \put(50.0,0.0){\fontsize{14.23}{17.07}\selectfont $q_{\ell + 1}$}
  \put(85.97,0.0){\fontsize{14.23}{17.07}\selectfont $2q_{\ell + 1}$}
  \put(220.97,0.0){\fontsize{14.23}{17.07}\selectfont $C_{\max}^{\pi}$}  
  \put(120.23,78.596){\fontsize{14.23}{17.07}\selectfont $\mcF$}
  \put(150.23,49.73){\fontsize{14.23}{17.07}\selectfont $\mcF$}
  \put(195.23,20.73){\fontsize{14.23}{17.07}\selectfont $\mcF$}
  \put(29.23,78.596){\fontsize{14.23}{17.07}\selectfont $J_{\ell + 3}, \ldots, J_n$}
  \put(65.23,51.73){\fontsize{14.23}{17.07}\selectfont $J_{\ell + 3}, \ldots, J_n$}
  \put(108.23,22.73){\fontsize{14.23}{17.07}\selectfont $J_{\ell + 3}, \ldots, J_n$}
  \put(193.23,78.596){\fontsize{14.23}{17.07}\selectfont $J_{\ell + 2}$}
  \put(25.23,50.73){\fontsize{14.23}{17.07}\selectfont $J_{\ell + 2}$}
  \put(61.23,22.73){\fontsize{14.23}{17.07}\selectfont $J_{\ell + 2}$}
  \put(163.23,79.596){\fontsize{14.23}{17.07}\selectfont $J_{\ell + 1}$}
  \put(195.23,49.73){\fontsize{14.23}{17.07}\selectfont $J_{\ell + 1}$}
  \put(28.23,22.73){\fontsize{14.23}{17.07}\selectfont $J_{\ell + 1}$}
  \end{picture}%
\caption{A feasible schedule $\pi$ for the $M3 \mid prpt \mid C_{\max}$ problem when there are at least two open-shop jobs and $p_1 < q_{\ell + 1}$.\label{fig41}}
\end{figure}
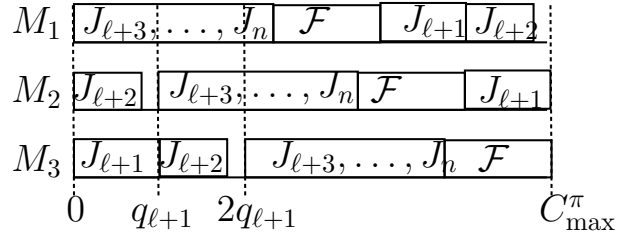

The makespan of the schedule $\pi$ is $C_{\max}^{\pi} = P(\mcF) + Q(\mcO) + q_{\ell + 1} - q_{\ell + 2}$.
Combining this with Eq.~(\ref{eq1}), we have
\[
C_{\max}^{\pi} \le P(\mcF) + Q(\mcO) + q_{\ell + 1} \le \frac{4}{3} C^*_{\max}.
\]

When there is only one open-shop job $J_n$ (that is, $\mcF = \{J_1, J_2, \ldots, J_{n-1}\}$ and $\mcO = \{J_n\}$),
construct a permutation schedule $\pi$ in which the job processing order for $M_1$ is $\langle \mcF, J_n\rangle$,
where the jobs of $\mcF$ are processed in the LPT order;
the job processing order for $M_2$ is $\langle \mcF, J_n\rangle$;
the job processing order for $M_3$ is $\langle J_n, \mcF\rangle$ (see for an illustration in Figure~\ref{fig42}).
If $P(\mcF) \le q_n$, which is shown in Figure~\ref{fig42}, then $\pi$ has makespan $3 q_n$ and thus is optimal.
\begin{figure}[ht]
\centering
  \setlength{\unitlength}{0.8bp}%
  \begin{picture}(305.20, 119.88)(0,0)
  \put(0,0){\includegraphics[scale=0.8]{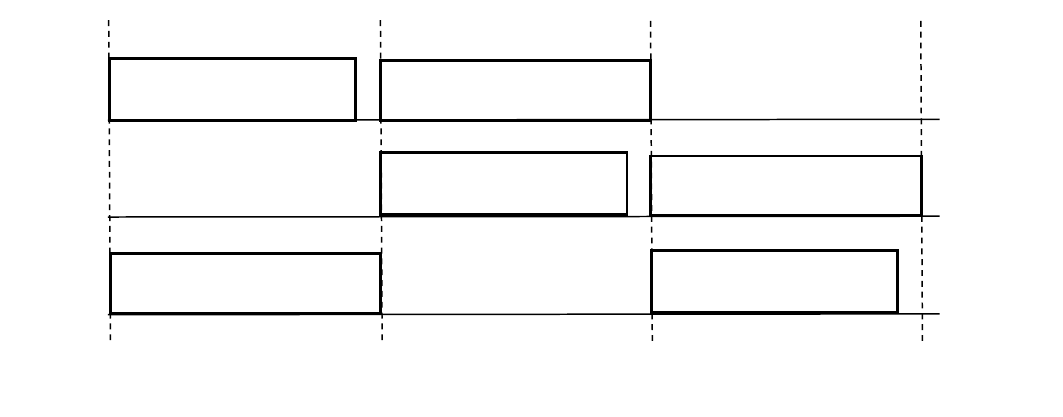}}
  \put(45.60,90.39){\fontsize{14.23}{17.07}\selectfont ${\cal F}$}
  \put(122.72,62.04){\fontsize{14.23}{17.07}\selectfont ${\cal F}$}
  \put(200.08,34.14){\fontsize{14.23}{17.07}\selectfont ${\cal F}$}
  \put(122.72,89.26){\fontsize{14.23}{17.07}\selectfont $J_n$}
  \put(200.98,61.13){\fontsize{14.23}{17.07}\selectfont $J_n$}
  \put(46.28,33.23){\fontsize{14.23}{17.07}\selectfont $J_n$}
  \put(5.67,90.17){\fontsize{14.23}{17.07}\selectfont $M_1$}
  \put(5.67,62.95){\fontsize{14.23}{17.07}\selectfont $M_2$}
  \put(5.67,35.72){\fontsize{14.23}{17.07}\selectfont $M_3$}
  \put(102.76,8.73){\fontsize{14.23}{17.07}\selectfont $q_n$}
  \put(179.89,8.73){\fontsize{14.23}{17.07}\selectfont $2q_n$}
  \put(257.01,8.73){\fontsize{14.23}{17.07}\selectfont $3q_n$}
  \put(25.64,8.73){\fontsize{14.23}{17.07}\selectfont $0$}
  \end{picture}%
\caption{A feasible schedule $\pi$ for the $M3 \mid prpt \mid C_{\max}$ problem when there is only one open-shop job $J_n$ and $p_1 < q_n$;
	the configuration shown here corresponds to $P(\mcF) \le q_n$ and thus $\pi$ is optimal.\label{fig42}}
\end{figure}
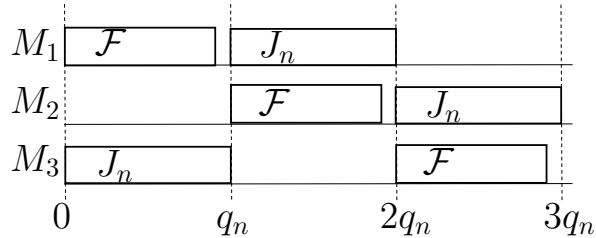
If $P(\mcF) > q_n$, then $\pi$ has makespan $C_{\max}^{\pi} \le 2 q_n + P(\mcF) \le \frac{4}{3} C^*_{\max}$ by Eq.~(\ref{eq1}).
This finishes the proof of the theorem.
\end{proof}

\begin{remark}
Construct an instance in which $p_i = \frac{1}{\ell - 1}$ for all $i = 1, 2, \ldots, \ell$,
$q_{\ell +1} = 1$ and $q_i = \frac{1}{n - \ell - 2}$ for all $i = \ell+2, \ell+3, \ldots, n$.
Then for this instance, the schedule $\pi$ constructed in the proof of Theorem \ref{thm41} has makespan $C_{\max}^{\pi} = 4 + \frac{1}{\ell - 1}$;
an optimal schedule has makespan $C_{\max}^* = 3 + \frac{1}{\ell - 1} + \frac{1}{n - \ell - 2}$ (see for an illustration in Figure~\ref{fig43}).
This suggests that the approximation ratio of $4/3$ is asymptotically tight for the algorithm presented in the proof of Theorem \ref{thm41}.
\end{remark}

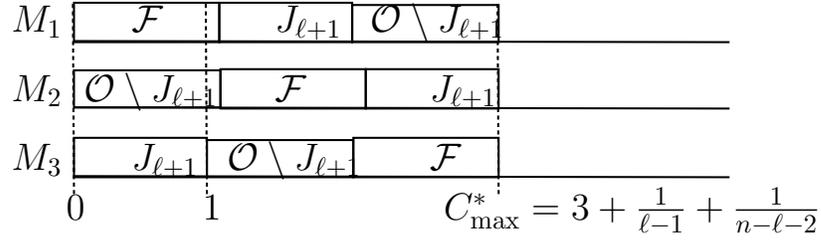
\begin{figure}[ht]
\centering
  \setlength{\unitlength}{0.9bp}%
  \begin{picture}(301.85, 99.05)(0,0)
  \put(20,5){\includegraphics[scale=0.9]{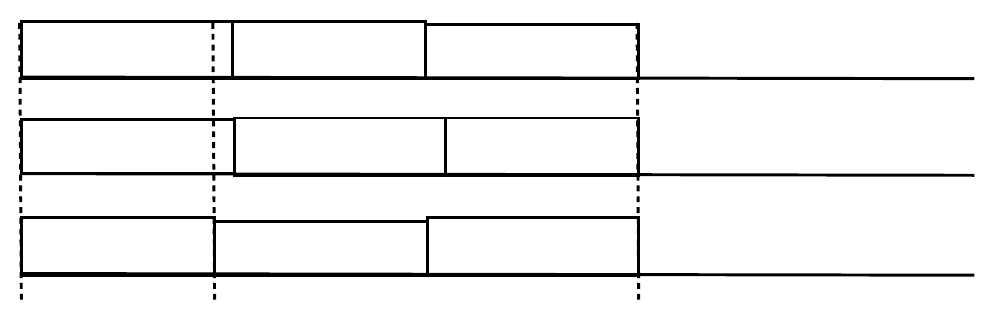}}
  \put(0.0,78.596){\fontsize{14.23}{17.07}\selectfont $M_1$}
  \put(0.0,49.73){\fontsize{14.23}{17.07}\selectfont $M_2$}
  \put(0.0,20.73){\fontsize{14.23}{17.07}\selectfont $M_3$}
  \put(23.0,0.0){\fontsize{14.23}{17.07}\selectfont $0$}
  \put(80.0,0.0){\fontsize{14.23}{17.07}\selectfont $1$}
  \put(180.97,0.0){\fontsize{14.23}{17.07}\selectfont $C_{\max}^* = 3 + \frac{1}{\ell - 1} + \frac{1}{n - \ell - 2}$}  
  \put(50.23,78.596){\fontsize{14.23}{17.07}\selectfont $\mcF$}
  \put(110.23,49.73){\fontsize{14.23}{17.07}\selectfont $\mcF$}
  \put(175.23,20.73){\fontsize{14.23}{17.07}\selectfont $\mcF$}
  \put(110.23,78.596){\fontsize{14.23}{17.07}\selectfont $J_{\ell + 1}$}
  \put(175.23,49.73){\fontsize{14.23}{17.07}\selectfont $J_{\ell + 1}$}
  \put(50.23,20.73){\fontsize{14.23}{17.07}\selectfont $J_{\ell + 1}$}
  \put(150.23,78.596){\fontsize{14.23}{17.07}\selectfont $\mcO \setminus J_{\ell + 1}$}
  \put(30.23,49.73){\fontsize{14.23}{17.07}\selectfont $\mcO \setminus J_{\ell + 1}$}
  \put(90.23,20.73){\fontsize{14.23}{17.07}\selectfont $\mcO \setminus J_{\ell + 1}$}
  \end{picture}%
\caption{An optimal schedule for the constructed instance of the $M3 \mid prpt \mid C_{\max}$ problem,
	in which $p_i = \frac{1}{\ell - 1}$ for all $i = 1, 2, \ldots, n$,
	$q_{\ell +1} = 1$ and $q_i = \frac{1}{n - \ell - 2}$ for all $i = \ell+2, \ell+3, \ldots, n$.\label{fig43}}
\end{figure}

\section{NP-hardness for $M3 \mid prpt, (n-1, 1) \mid C_{\max}$ when $p_1 < q_n$}
\label{sec4}
%==================================================================================================
Recall that we use $M3 \mid prpt, (n-1, 1) \mid C_{\max}$ to denote the special of $M3 \mid prpt \mid C_{\max}$ where there is only one open-shop job,
{\it i.e.}, $\mcF = \{J_1, J_2, \ldots, J_{n-1}\}$ and $\mcO = \{J_n\}$.
In this section,
we show that this special case $M3 \mid prpt, (n-1, 1) \mid C_{\max}$ is already NP-hard if the unique open-shop job is larger than any flow-shop job, {\it i.e.},
$p_1 < q_n$.
We prove the NP-hardness through a reduction from the {\sc Partition} problem \cite{GJ79}, which is a well-known NP-complete problem.

\begin{theorem}
\label{thm51}
The $M3 \mid prpt, (n-1, 1) \mid C_{\max}$ problem is NP-hard if the unique open-shop job is larger than any flow-shop job.
\end{theorem}
\begin{proof}
An instance of the {\sc Partition} problem consists of
a set $S = \{ a_1, a_2, a_3, \ldots, a_m \}$ where each $a_i$ is a positive integer and $a_1 + a_2 + \ldots + a_m = 2B$,
and the query is whether or not $S$ can be partitioned into two parts such that each part sums to exactly $B$.

Let $x > B$, and we assume that $a_1 \ge a_2 \ge \ldots \ge a_m$.

We construct an instance of the $M3 \mid prpt \mid C_{\max}$ problem as follows:
there are in total $m+2$ flow-shop jobs, and their processing times are $p_1 = x, p_2 = x$, and $p_{i+2} = a_i$ for $i = 1, 2, \ldots, m$;
there is only one open-shop job with processing time $q_{m+3} = B + 2x$.
Note that the total number of jobs is $n = m+3$, and one sees that the open-shop job is larger than any flow-shop job.

If the set $S$ can be partitioned into two parts $S_1$ and $S_2$ such that each part sums to exactly $B$,
then we let $\mcJ^1 = J_1 \cup \{J_i \mid a_i \in S_1\}$ and $\mcJ^2 = J_2 \cup \{J_i \mid a_i \in S_2\}$.
We construct a permutation schedule $\pi$ in which the job processing order for $M_1$ is $\langle \mcJ^1, \mcJ^2, J_{m+3}\rangle$,
where the jobs of $\mcJ^1$ and the jobs of $\mcJ^2$ are processed in the LPT order, respectively;
the job processing order for $M_2$ is $\langle \mcJ^1, J_{m+3}, \mcJ^2\rangle$;
the job processing order for $M_3$ is $\langle J_{m+3}, \mcJ^1, \mcJ^2\rangle$.
See Figure~\ref{fig51} for an illustration, in which $J_1$ starts at time $0$ on $M_1$, starts at time $x$ on $M_2$, and starts at time $B+2x$ on $M_3$;
$J_2$ starts at time $B+x$ on $M_1$, starts at time $2B+4x$ on $M_2$, and starts at time $2B+5x$ on $M_3$;
$J_{m+3}$ starts at time $0$ on $M_3$, starts at time $B+2x$ on $M_2$, and starts at time $2B+4x$ on $M_1$.
The feasibility is trivial and its makespan is $C^{\pi}_{\max} = 3q_{m+3} = 3B + 6x$, suggesting the optimality.

\begin{figure}[ht]
\centering
  \setlength{\unitlength}{0.9bp}%
  \begin{picture}(301.85, 99.05)(0,0)
  \put(20,5){\includegraphics[scale=0.9]{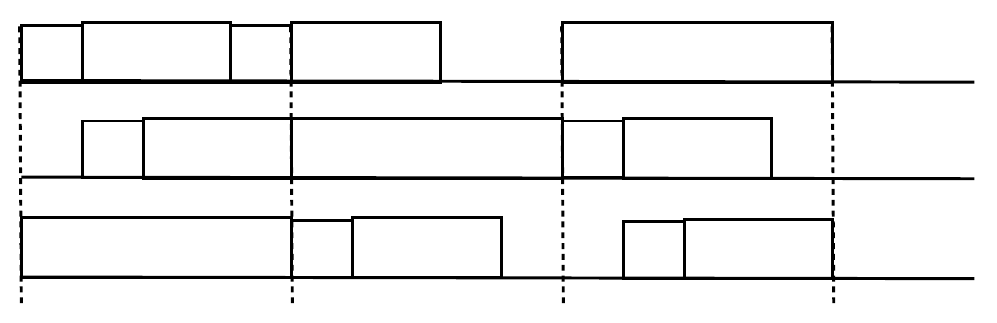}}
  \put(0.0,78.596){\fontsize{14.23}{17.07}\selectfont $M_1$}
  \put(0.0,49.73){\fontsize{14.23}{17.07}\selectfont $M_2$}
  \put(0.0,20.73){\fontsize{14.23}{17.07}\selectfont $M_3$}
  \put(23.0,0.0){\fontsize{14.23}{17.07}\selectfont $0$}
  \put(90.0,0.0){\fontsize{14.23}{17.07}\selectfont $B + 2x$}
  \put(165.97,0.0){\fontsize{14.23}{17.07}\selectfont $2B + 4x$}
  \put(250.97,0.0){\fontsize{14.23}{17.07}\selectfont $3B + 6x$}  
  \put(28.23,78.596){\fontsize{14.23}{17.07}\selectfont $J_1$}
  \put(45.23,49.73){\fontsize{14.23}{17.07}\selectfont $J_1$}
  \put(105.23,20.73){\fontsize{14.23}{17.07}\selectfont $J_1$}
  \put(88.23,78.596){\fontsize{14.23}{17.07}\selectfont $J_2$}
  \put(185.23,49.73){\fontsize{14.23}{17.07}\selectfont $J_2$}
  \put(200.23,20.73){\fontsize{14.23}{17.07}\selectfont $J_2$}
  \put(208.23,78.596){\fontsize{14.23}{17.07}\selectfont $J_{m+3}$}
  \put(135.23,49.73){\fontsize{14.23}{17.07}\selectfont $J_{m+3}$}
  \put(50.23,20.73){\fontsize{14.23}{17.07}\selectfont $J_{m+3}$} 
  \put(58.23,78.596){\fontsize{14.23}{17.07}\selectfont $\mcJ^1$}
  \put(75.23,49.73){\fontsize{14.23}{17.07}\selectfont $\mcJ^1$}
  \put(140.23,20.73){\fontsize{14.23}{17.07}\selectfont $\mcJ^1$} 
  \put(118.23,78.596){\fontsize{14.23}{17.07}\selectfont $\mcJ^2$}
  \put(215.23,49.73){\fontsize{14.23}{17.07}\selectfont $\mcJ^2$}
  \put(235.23,20.73){\fontsize{14.23}{17.07}\selectfont $\mcJ^2$} 
  \end{picture}%
\caption{A feasible schedule $\pi$ for the constructed instance of the $M3 \mid prpt \mid C_{\max}$ problem,
	when the set $S$ can be partitioned into two equal parts $S_1$ and $S_2$.
	The partition of the flow-shop jobs $\{\mcJ^1, \mcJ^2\}$ is correspondingly constructed.
	In the schedule, the jobs of $\mcJ^1$ and the jobs of $\mcJ^2$ are processed in the LPT order, respectively.\label{fig51}}
\end{figure}
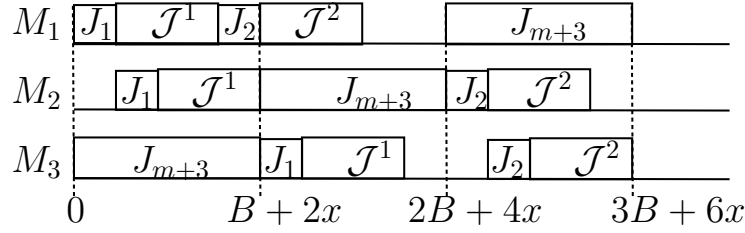

Conversely, if the optimal makespan for the constructed instance is $C_{\max}^* = 3B + 6x = 3 q_{m+3}$,
then we will show next that $S$ admits a partition into two equal parts.

Firstly, we see that the second machine processing the open-shop job $J_{m+3}$ cannot be $M_1$,
since otherwise $M_1$ has to process all the jobs of $\mcF$ before $J_{m+3}$, leading to a makespan larger than $3B + 6x$;
the second machine processing the open-shop job $J_{m+3}$ cannot be $M_3$ either,
since otherwise $M_3$ has no room to process any job of $\mcF$ before $J_{m+3}$, leading to a makespan larger than $3B + 6x$ too.
Therefore, the second machine processing the open-shop job $J_{m+3}$ has to be $M_2$, see Figure~\ref{fig52} for an illustration.

\begin{figure}[ht]
\centering
  \setlength{\unitlength}{0.9bp}%
  \begin{picture}(301.85, 99.05)(0,0)
  \put(20,5){\includegraphics[scale=0.9]{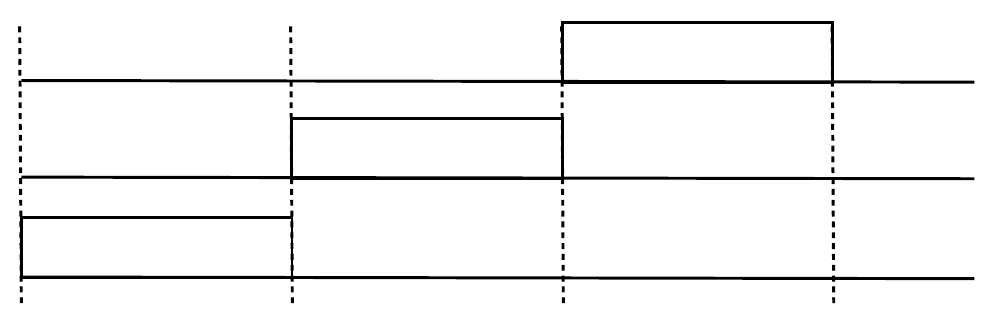}}
  \put(0.0,78.596){\fontsize{14.23}{17.07}\selectfont $M_1$}
  \put(0.0,49.73){\fontsize{14.23}{17.07}\selectfont $M_2$}
  \put(0.0,20.73){\fontsize{14.23}{17.07}\selectfont $M_3$}
  \put(23.0,0.0){\fontsize{14.23}{17.07}\selectfont $0$}
  \put(90.0,0.0){\fontsize{14.23}{17.07}\selectfont $B + 2x$}
  \put(165.97,0.0){\fontsize{14.23}{17.07}\selectfont $2B + 4x$}
  \put(250.97,0.0){\fontsize{14.23}{17.07}\selectfont $3B + 6x$}  
  \put(208.23,78.596){\fontsize{14.23}{17.07}\selectfont $J_{m+3}$}
  \put(135.23,49.73){\fontsize{14.23}{17.07}\selectfont $J_{m+3}$}
  \put(50.23,20.73){\fontsize{14.23}{17.07}\selectfont $J_{m+3}$} 
  \put(75.23,49.73){\fontsize{14.23}{17.07}\selectfont $\mcF^1$}
  \put(225.23,49.73){\fontsize{14.23}{17.07}\selectfont $\mcF^2$}
  \end{picture}%
\caption{An illustration of an optimal schedule for the constructed instance of the $M3 \mid prpt, (n-1, 1) \mid C_{\max}$ problem
	with $\mcO = \{J_{m+3}\}$ and $q_{m+3} = B + 2x$.
	Its makespan is $C_{\max}^* = 3B + 6x = 3 q_{m+3}$.\label{fig52}}
\end{figure}
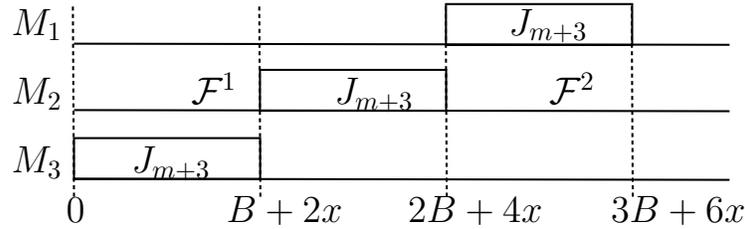

Denote the job subsets processed before and after the job $J_{m+3}$ on $M_2$ as $\mcF^1$ and $\mcF^2$, respectively.
Since $x > B$, neither of $\mcF^1$ and $\mcF^2$ may contain both $J_1$ and $J_2$, which have processing times $x$.
It follows that $\mcF^1$ and $\mcF^2$ each contains exactly one of $J_1$ and $J_2$, and subsequently $P(\mcF^1) = P(\mcF^2) = B + x$.
Therefore, the jobs of $\mcF^1 \setminus \{J_1, J_2\}$ have a total processing time of exactly $B$, suggesting a subset of $S$ sums to exactly $B$.
This finishes the proof of the theorem. 
\end{proof}

\section{An FPTAS for $M3 \mid prpt, (n-1, 1) \mid C_{\max}$}
\label{sec6}
%==================================================================================================
Recall that the FPTAS designed in Theorem~\ref{thm37} designed for the general $M3 \mid prpt, (n-1, 1) \mid C_{\max}$ problem when $p_1 \ge q_n$
also works for the special case $M3 \mid prpt, (n-1, 1) \mid C_{\max}$ when $p_1 \ge q_n$.
In this section, we design another FPTAS for the $M3 \mid prpt, (n-1, 1) \mid C_{\max}$ problem when $p_1 < q_n$,
which is denoted as $M3 \mid prpt, (n-1, 1), p_1 < q_n \mid C_{\max}$ for simplicity.
That is, we design a $(1 + \epsilon)$-approximation algorithm $C(\epsilon)$ for this case, for any given $\epsilon > 0$,
and its running time is polynomial in both $n$ and $1/\epsilon$.

\subsection{A polynomial-time solvable case}
%--------------------------------------------------------------------------------------------------
The following lemma states that if the total processing time of all the flow-shop jobs is no greater than $q_n$,
then we can easily construct an optimal schedule in linear time.

\begin{lemma}
\label{lemma61}
If $P(\mcF) \le q_n$,
then the $M3 \mid prpt, (n-1, 1), p_1 < q_n \mid C_{\max}$ problem is solvable in linear time and $C_{\max}^* = 3q_n$. 
\end{lemma}
\begin{proof}
In this case, we construct a permutation schedule $\pi$ in which the job processing order for $M_1$ is $\langle \mcF, J_n\rangle$,
where the jobs of $\mcF$ are processed in the given (arbitrary, no need to be sorted) order;
the job processing order for $M_2$ is $\langle \mcF, J_n\rangle$;
the job processing order for $M_3$ is $\langle J_n, \mcF\rangle$.
As depicted in Figure~\ref{fig42}, the jobs of $\mcF$ are processed consecutively on each machine,
starting at time $0$ on $M_1$,
starting at time $q_n$ on $M_2$, and
starting at time $2 q_n$ on $M_3$;
the unique open-shop job $J_n$
starts at time $0$ on $M_3$,
starts at time $q_n$ on $M_1$, and
starts at time $2q_n$ on $M_2$.
The schedule is feasible due to $P(\mcF) \le q_n$ and its makespan is $3q_n$, and thus by Eq.~(\ref{eq1}) it is an optimal schedule.
\end{proof}

\subsection{Structural properties of optimal schedules}
%--------------------------------------------------------------------------------------------------
We assume in the rest of the section that $P(\mcF) > q_n$, from which we conclude that $\mcF$ contains at least two jobs.
We explore the structural properties of optimal schedules for designing our approximation algorithm.

\begin{lemma}
\label{lemma62}
There exists an optimal schedule for the $M3 \mid prpt, (n-1, 1), p_1 < q_n \mid C_{\max}$ problem
in which the open-shop job $J_n$ is processed on $M_3$ before it is processed on $M_1$.
\end{lemma}
\begin{proof}
We prove the lemma by construction.
Suppose $\pi^*$ is an optimal schedule in which the open-shop job $J_n$ is processed on $M_3$ after it is processed on $M_1$.
(That is, the machine order for $J_n$ is $\langle -, M_1, -, M_3, - \rangle$, with exactly one of the $-$'s replaced by $M_2$.)
Recall that we use $S_j^i$ ($C_j^i$, respectively) to denote the start (finish, respectively) time of the job $J_j$ on $M_i$ in $\pi^*$.
Clearly, $C_n^1 \le S_n^3$.

We determine in the following two time points $t_1$ and $t_2$.

We see that if $J_n$ is first processed on $M_2$, {\it i.e.}, $C_n^2 \le S_n^1$,
then all the jobs of $\mcF$ processed before $J_n$ on $M_2$ can be finished on $M_3$ by time $C_n^2$, due to $p_1 < q_n$.
It follows that if $J_n$ is first processed on $M_2$, then we may assume without loss of generality that in $\pi^*$,
there is no job $J_j$ such that $S_j^3 < C_n^2 < C_j^3$ (that is, the processing periods of $J_j$ on $M_3$ and of $J_n$ on $M_2$ intersect at time $C_n^2$).
%for any job $J_j$, either $C_j^3 \le C_n^2$ or $C_n^2 \le S_j^3$.
If there is a job $J_j$ such that $S_j^3 < S_n^1 < C_j^3$ (that is, the processing periods of $J_j$ on $M_3$ and of $J_n$ on $M_1$ intersect at time $S_n^1$),
then we set $t_1 = S_j^3$, otherwise we set $t_1 = S_n^1$ (see Figure~\ref{fig61} for an illustration).

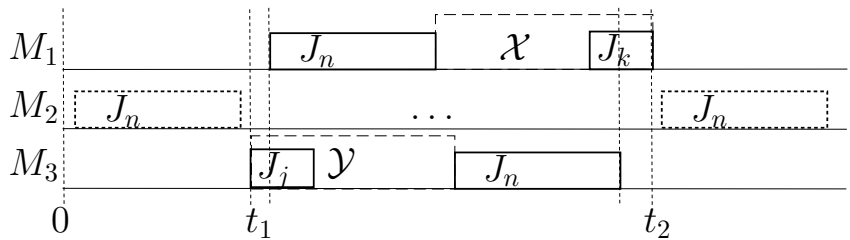
\begin{figure}[ht]
\centering
  \setlength{\unitlength}{0.8bp}%
  \begin{picture}(406.45, 120.12)(0,0)
  \put(0,0){\includegraphics[scale=0.8]{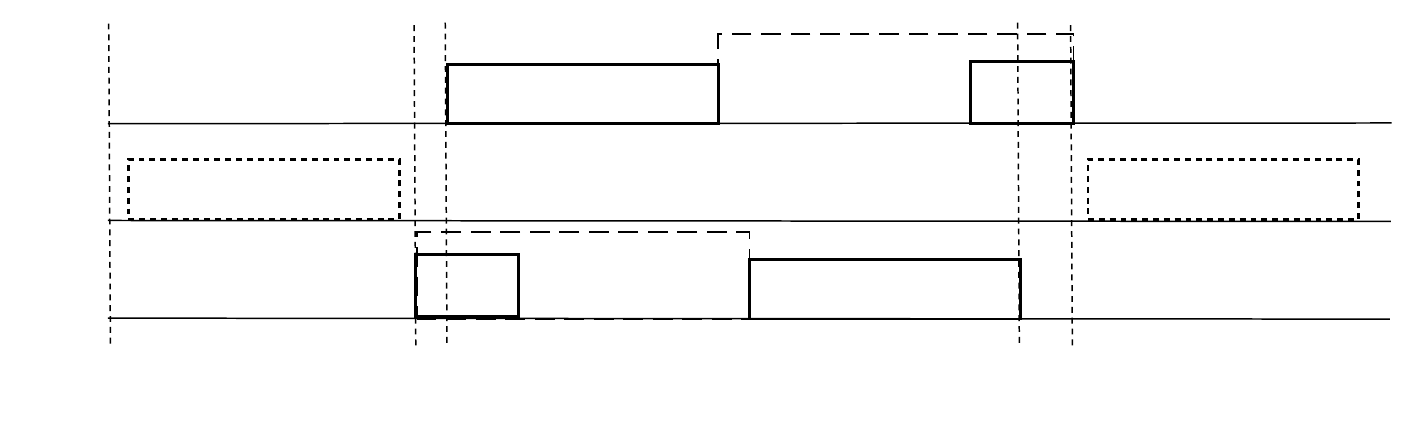}}
  \put(122.50,34.14){\fontsize{14.23}{17.07}\selectfont $J_j$}
  \put(142.12,89.26){\fontsize{14.23}{17.07}\selectfont $J_n$}
  \put(50.68,61.13){\fontsize{14.23}{17.07}\selectfont $J_n$}
  \put(5.67,90.17){\fontsize{14.23}{17.07}\selectfont $M_1$}
  \put(5.67,62.95){\fontsize{14.23}{17.07}\selectfont $M_2$}
  \put(5.67,35.72){\fontsize{14.23}{17.07}\selectfont $M_3$}
  \put(117.31,8.73){\fontsize{14.23}{17.07}\selectfont $t_1$}
  \put(25.64,8.73){\fontsize{14.23}{17.07}\selectfont $0$}
  \put(282.26,89.65){\fontsize{14.23}{17.07}\selectfont $J_k$}
  \put(229.15,33.02){\fontsize{14.23}{17.07}\selectfont $J_n$}
  \put(327.04,61.13){\fontsize{14.23}{17.07}\selectfont $J_n$}
  \put(304.89,8.73){\fontsize{14.23}{17.07}\selectfont $t_2$}
  \put(194.88,62.06){\fontsize{14.23}{17.07}\selectfont $\ldots$}
  \put(236.62,89.65){\fontsize{14.23}{17.07}\selectfont ${\cal X}$}
  \put(155.79,36.20){\fontsize{14.23}{17.07}\selectfont ${\cal Y}$}
  \end{picture}%
\caption{Two possible values for $t_1$: if the job $J_j$ as shown exists, then $t_1 = S_j^3$, otherwise $t_1 = S_n^1$;
	symmetrically, two possible values for $t_2$: if the job $J_k$ as shown exists, then $t_2 = C_k^1$, otherwise $t_2 = C_n^3$.\label{fig61}}
\end{figure}

Symmetrically, if $J_n$ is last processed on $M_2$, {\it i.e.}, $C_n^3 \le S_n^2$,
then all the jobs of $\mcF$ processed after $J_n$ on $M_2$ can be started on $M_1$ by time $S_n^2$, due to $p_1 < q_n$.
It follows that if $J_n$ is last processed on $M_2$, then we may assume without loss of generality that in $\pi^*$,
there is no job $J_k$ such that $S_k^1 < S_n^2 < C_k^1$ (that is, the processing periods of $J_k$ on $M_1$ and of $J_n$ on $M_2$ intersect at time $S_n^2$).
%for any job $J_k$, either $C_k^1 \le S_n^2$ or $S_n^2 \le S_k^1$.
If there is a job $J_k$ such that $S_k^1 < C_n^3 < C_k^1$ (that is, the processing periods of $J_k$ on $M_1$ and of $J_n$ on $M_3$ intersect at time $C_n^3$),
then we set $t_2 = C_k^1$, otherwise we set $t_2 = C_n^3$ (see Figure~\ref{fig61} for an illustration).

We note that both $t_1$ and $t_2$ are well defined.
On the machine $M_2$, the job $J_n$ is either finished before time $t_1$ ({\it i.e.}, $C_n^2 \le t_1$),
or started after time $t_2$ ({\it i.e.}, $t_2 \le S_n^2$),
or is processed in between the closed time interval $[C_n^1, S_n^3]$.
We thus obtain from $\pi^*$ another schedule $\pi$ by 
1) moving the job subset $\mcX$ originally processed on $M_1$ in between the closed time interval $[C_n^1, t_2]$ to start exactly $q_n$ time units earlier,
	while moving $J_n$ on $M_1$ to start at time $t_2 - q_n$, and
2) moving the job subset $\mcY$ originally processed on $M_3$ in between the closed time interval $[t_1, S_n^3]$ to start exactly $q_n$ time units later,
	while moving $J_n$ on $M_3$ to start at time $t_1$.
See for an illustration in Figure~\ref{fig62} and how it is obtained from the configuration in Figure~\ref{fig61}.
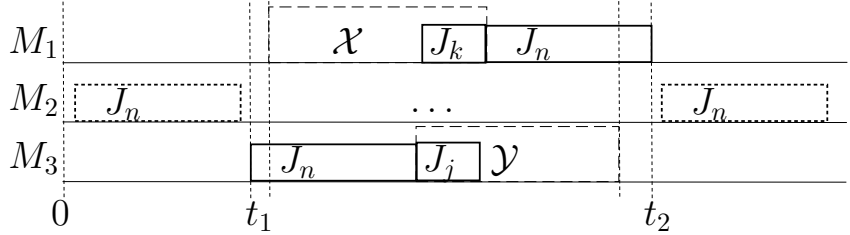
\begin{figure}[ht]
\centering
  \setlength{\unitlength}{0.8bp}%
  \begin{picture}(406.45, 120.12)(0,0)
  \put(0,0){\includegraphics[scale=0.8]{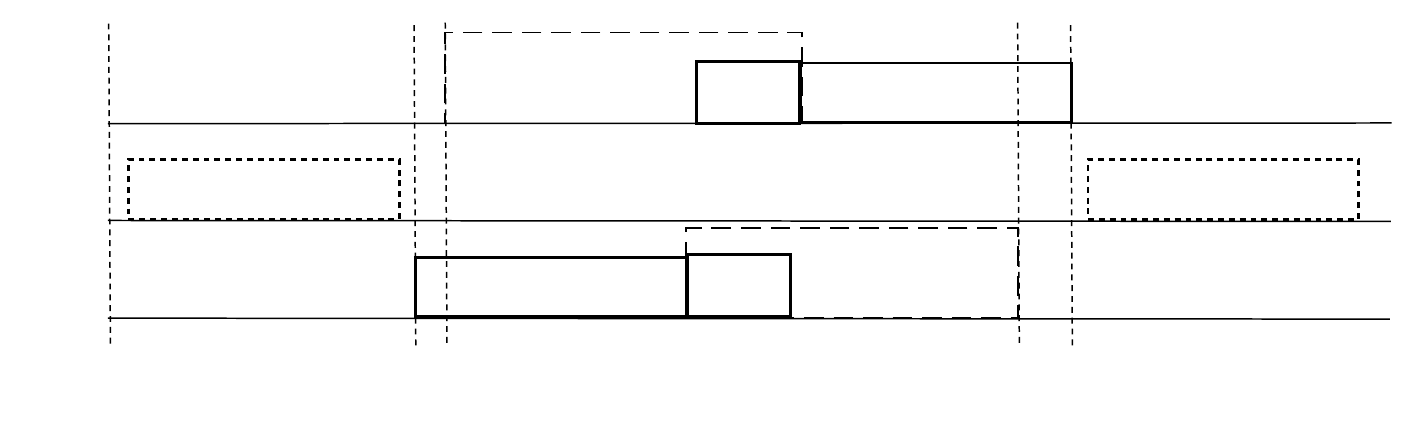}}
  \put(200.73,34.14){\fontsize{14.23}{17.07}\selectfont $J_j$}
  \put(243.81,89.58){\fontsize{14.23}{17.07}\selectfont $J_n$}
  \put(50.68,61.13){\fontsize{14.23}{17.07}\selectfont $J_n$}
  \put(5.67,90.17){\fontsize{14.23}{17.07}\selectfont $M_1$}
  \put(5.67,62.95){\fontsize{14.23}{17.07}\selectfont $M_2$}
  \put(5.67,35.72){\fontsize{14.23}{17.07}\selectfont $M_3$}
  \put(117.31,8.73){\fontsize{14.23}{17.07}\selectfont $t_1$}
  \put(25.64,8.73){\fontsize{14.23}{17.07}\selectfont $0$}
  \put(203.38,89.65){\fontsize{14.23}{17.07}\selectfont $J_k$}
  \put(132.99,33.67){\fontsize{14.23}{17.07}\selectfont $J_n$}
  \put(327.04,61.13){\fontsize{14.23}{17.07}\selectfont $J_n$}
  \put(304.89,8.73){\fontsize{14.23}{17.07}\selectfont $t_2$}
  \put(194.88,62.06){\fontsize{14.23}{17.07}\selectfont $\ldots$}
  \put(157.74,89.65){\fontsize{14.23}{17.07}\selectfont ${\cal X}$}
  \put(232.71,34.24){\fontsize{14.23}{17.07}\selectfont ${\cal Y}$}
  \end{picture}%
\caption{The schedule $\pi$ obtained from the optimal schedule $\pi^*$ by
	swapping the job subset $\mcX$ processed on $M_1$ in between the closed time interval $[C_n^1, t_2]$ with $J_n$ on $M_1$, and
	swapping the job subset $\mcY$ processed on $M_3$ in between the closed time interval $[t_1, S_n^3]$ with $J_n$ on $M_3$.\label{fig62}}
\end{figure}
The schedule $\pi$ is feasible because the processing times of $J_n$ on all three machines are still non-overlapping and
all the other jobs are flow-shop jobs which can only be processed earlier on $M_1$ and/or be processed later on $M_3$.
Since no other job is moved, one clearly sees that $\pi$ maintains the same makespan as $\pi^*$ and thus $\pi$ is also an optimal schedule,
in which $J_n$ is processed on $M_3$ before it is processed on $M_1$.
This finishes the proof of the lemma.
\end{proof}

From Lemma \ref{lemma62}, we conclude that in the optimal schedules the machine order for $J_n$ is $\langle -, M_3, -, M_1, - \rangle$,
with exactly one of the $-$'s replaced by $M_2$.
The following two lemmas discuss the optimal schedules constrained to the machine order for $J_n$.

\begin{lemma}
\label{lemma63}
If the machine order for $J_n$ is forced to be $\langle M_2, M_3, M_1\rangle$ or $\langle M_3, M_1, M_2\rangle$,
then an optimal schedule $\pi^*$ can be constructed in $O(n \log n)$ time and its makespan is $C_{\max}^* = 2q_n + P(\mcF)$.
\end{lemma}
\begin{proof}
We prove the lemma for the case where the machine order for $J_n$ is $\langle M_2, M_3, M_1\rangle$;
the case where the machine order for $J_n$ is $\langle M_3, M_1, M_2\rangle$ can be almost identically argued, by swapping $M_3$ with $M_1$.

Let $\pi$ be a schedule in which the machine order for $J_n$ is $\langle M_2, M_3, M_1\rangle$.
Since $J_n$ is last processed by $M_1$ while all the other jobs are flow-shop jobs,
we may swap $J_n$ with the job subset processed on $M_1$ after $J_n$, if any.
This swapping certainly does not increase the makespan and thus we may assume that, on $M_1$, all the jobs of $\mcF$ are processed before $J_n$.
We prove the lemma by showing that $C_{\max}^{\pi} \ge 2q_n + P(\mcF)$ and constructing a concrete schedule with its makespan equal to $2q_n + P(\mcF)$.
Recall that we use $S_j^i$ ($C_j^i$, respectively) to denote the start (finish, respectively) time of the job $J_j$ on $M_i$ in $\pi$.

Similar to the place where we determine the time point $t_1$ in the proof of Lemma~\ref{lemma62},
all the jobs of $\mcF$ processed before $J_n$ on $M_2$, which form the subset $\mcX$, can be finished on $M_3$ by time $C_n^2$, due to $p_1 < q_n$.
It follows that $J_n$ may immediately start on $M_3$ once it is finished on $M_2$, {\it i.e.}, $S_n^3 = C_n^2$,
while all the other jobs of $\mcF$ not processed by time $C_n^2$ on $M_3$, which form the subset $\mcY$,
can be moved to be processed after $J_n$ (in their original processing order on $M_3$ in $\pi$).
See for an illustration in Figure~\ref{fig63}.
Since $\mcX \cup \mcY = \mcF$, the completion time for the machine $M_3$ is at least $P(\mcX) + q_n + q_n + P(\mcY) \ge 2q_n + P(\mcF)$,
and consequently $C_{\max}^{\pi} \ge 2q_n + P(\mcF)$.

\begin{figure}[ht]
\centering
  \setlength{\unitlength}{0.8bp}%
  \begin{picture}(335.29, 119.88)(0,0)
  \put(0,0){\includegraphics[scale=0.8]{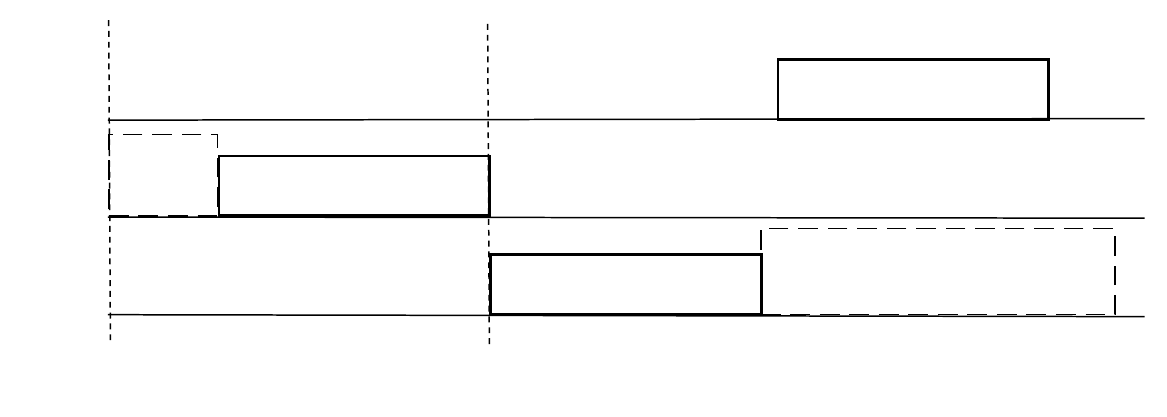}}
  \put(237.29,89.58){\fontsize{14.23}{17.07}\selectfont $J_n$}
  \put(76.76,61.13){\fontsize{14.23}{17.07}\selectfont $J_n$}
  \put(5.67,90.17){\fontsize{14.23}{17.07}\selectfont $M_1$}
  \put(5.67,62.95){\fontsize{14.23}{17.07}\selectfont $M_2$}
  \put(5.67,35.72){\fontsize{14.23}{17.07}\selectfont $M_3$}
  \put(25.64,8.73){\fontsize{14.23}{17.07}\selectfont $0$}
  \put(154.50,33.34){\fontsize{14.23}{17.07}\selectfont $J_n$}
  \put(36.49,62.27){\fontsize{14.23}{17.07}\selectfont ${\cal X}$}
  \put(236.95,35.87){\fontsize{14.23}{17.07}\selectfont ${\cal Y}$}
  \end{picture}%
\caption{The schedule $\pi$ in which the machine order for $J_n$ is $\langle M_2, M_3, M_1\rangle$.
	The job $J_n$ starts on $M_3$ immediately after it is finished on $M_2$.
	The job subset $\mcX$ is processed on $M_2$ before $J_n$ and the job subset $\mcY$ is processed on $M_3$ after $J_n$.\label{fig63}}
\end{figure}
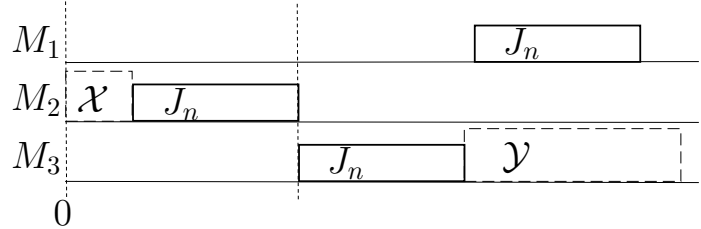

Clearly, if $\mcX = \emptyset$ in the schedule $\pi$, and the jobs of $\mcF$ are processed on all the three machines in the same LPT order,
then the completion times for $M_1, M_2, M_3$ are $\max\{P(\mcF), 2q_n\} + q_n$, $q_n + P(\mcF)$, and $2q_n + P(\mcF)$, respectively.
From the fact that $P(\mcF) > q_n$, we conclude that the makespan of such a schedule is exactly $2q_n + P(\mcF)$.
Note that the schedule can be constructed in $O(n \log n)$ time.
This finishes the proof of the lemma.
\end{proof}

\begin{lemma}
\label{lemma64}
If the machine order for $J_n$ is forced to be $\langle M_3, M_2, M_1\rangle$,
and $\mcF^1, \mcF^2 \subseteq \mcF$ are the subsets of jobs forced to be processed before and after $J_n$ on the machine $M_2$, respectively,
then an optimal schedule $\pi^*$ can be constructed in $O(n \log n)$ time and its makespan is
$C_{\max}^* = \max\{p_{i_1} + P(\mcF^1), q_n\} + q_n + \max\{p_{i_2} + P(\mcF^2), q_n\}$,
where $J_{i_1}$ is the largest job of $\mcF^1$ and $J_{i_2}$ is the largest job of $\mcF^2$.
\end{lemma}
\begin{proof}
Consider a feasible $\pi$ in which the machine order for $J_n$ is $\langle M_3, M_2, M_1\rangle$,
and $\mcF^1, \mcF^2 \subseteq \mcF$ are the subsets of jobs processed before and after $J_n$ on $M_2$, respectively.

Similarly as in the proof of Lemma~\ref{lemma63}, since $J_n$ is last processed by $M_1$ while all the other jobs are flow-shop jobs,
we may swap $J_n$ with the job subset processed on $M_1$ after $J_n$, if any, in the schedule $\pi$.
This swapping certainly does not increase the makespan and thus we may assume that, on $M_1$, all the jobs of $\mcF$ are processed before $J_n$.
Symmetrically, we may also assume that, on $M_3$, all the jobs of $\mcF$ are processed after $J_n$.

Note that $\mcF^1, \mcF^2 \subseteq \mcF$ are the subsets of jobs processed before and after $J_n$ in the schedule $\pi$, respectively.
We can also assume that $\mcF^1 \ne \emptyset$, as otherwise we may swap to process $J_n$ first on $M_2$ and then to process $J_n$ on $M_3$
to convert $\pi$ into a feasible schedule for the case where the machine order for $J_n$ is $\langle M_2, M_3, M_1\rangle$.
It follows that the optimal schedule $\pi^*$ constructed in $O(n \log n)$ time in Lemma~\ref{lemma63} serves,
with its makespan $C_{\max}^* = 2q_n + p(\mcF)$.
For the same reason, we can assume that $\mcF^2 \ne \emptyset$.

Since all the jobs of $\mcF^1$ have to be finished by time $S_n^2$ on the two machines $M_1$ and $M_2$,
we have $S_n^2 \ge p_{i_1} + P(\mcF^1)$.
From $C_n^3 \le S_n^2$, we conclude that $S_n^2 \ge \max\{p_{i_1} + P(\mcF^1), q_n\}$.
Similarly, since all the jobs of $\mcF^2$ have to be processed on the two machines $M_2$ and $M_3$, and the earliest starting time is $C_n^2$,
the completion time for the machine $M_3$ is at least $C_n^2 + p_{i_2} + P(\mcF^2)$.
Note that the earliest starting time for processing $J_n$ on $M_1$ is $C_n^2$ too.
It follows that the makespan of the schedule $\pi$ is $C_{\max}^{\pi} \ge \max\{p_{i_1} + P(\mcF^1), q_n\} + q_n + \max\{p_{i_2} + P(\mcF^2), q_n\}$.
See for an illustration in Figure~\ref{fig64}.

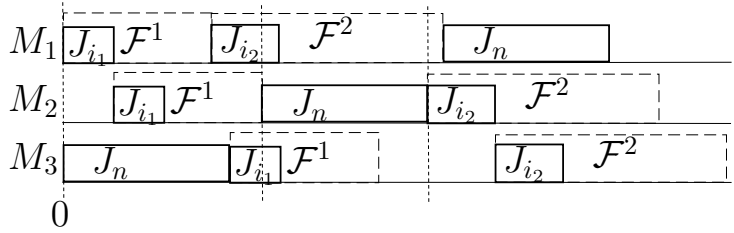
\begin{figure}[ht]
\centering
  \setlength{\unitlength}{0.8bp}%
  \begin{picture}(360.56, 119.88)(0,0)
  \put(0,0){\includegraphics[scale=0.8]{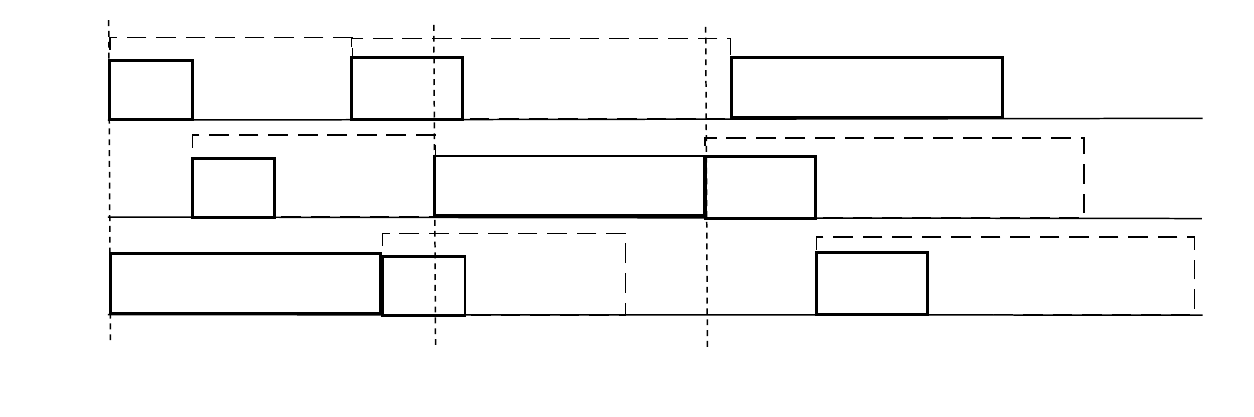}}
  \put(223.93,90.12){\fontsize{14.23}{17.07}\selectfont $J_n$}
  \put(138.76,61.13){\fontsize{14.23}{17.07}\selectfont $J_n$}
  \put(5.67,90.17){\fontsize{14.23}{17.07}\selectfont $M_1$}
  \put(5.67,62.95){\fontsize{14.23}{17.07}\selectfont $M_2$}
  \put(5.67,35.72){\fontsize{14.23}{17.07}\selectfont $M_3$}
  \put(25.64,8.73){\fontsize{14.23}{17.07}\selectfont $0$}
  \put(44.94,33.61){\fontsize{14.23}{17.07}\selectfont $J_n$}
  \put(58.14,90.06){\fontsize{14.23}{17.07}\selectfont ${\cal F}^1$}
  \put(281.04,36.40){\fontsize{14.23}{17.07}\selectfont ${\cal F}^2$}
  \put(33.20,88.92){\fontsize{14.23}{17.07}\selectfont $J_{i_1}$}
  \put(81.92,62.00){\fontsize{14.23}{17.07}\selectfont ${\cal F}^1$}
  \put(56.98,60.86){\fontsize{14.23}{17.07}\selectfont $J_{i_1}$}
  \put(238.71,34.14){\fontsize{14.23}{17.07}\selectfont $J_{i_2}$}
  \put(147.16,92.79){\fontsize{14.23}{17.07}\selectfont ${\cal F}^2$}
  \put(104.82,90.53){\fontsize{14.23}{17.07}\selectfont $J_{i_2}$}
  \put(248.98,64.20){\fontsize{14.23}{17.07}\selectfont ${\cal F}^2$}
  \put(206.64,61.93){\fontsize{14.23}{17.07}\selectfont $J_{i_2}$}
  \put(136.71,33.68){\fontsize{14.23}{17.07}\selectfont ${\cal F}^1$}
  \put(111.77,32.54){\fontsize{14.23}{17.07}\selectfont $J_{i_1}$}
  \end{picture}%
\caption{The schedule $\pi$ in which the machine order for $J_n$ is $\langle M_3, M_2, M_1\rangle$.
	The job subsets $\mcF^1$ and $\mcF^2$ are processed on $M_2$ before and after $J_n$, respectively.
	The jobs $J_{i_1}$ and $J_{i_2}$ are the largest job of $\mcF^1$ and $\mcF^2$, respectively.\label{fig64}}
\end{figure}

On the other hand, if the jobs of $\mcF^1$ are processed in the same LPT order on $M_1$ and $M_2$ in the schedule $\pi$,
then they are finished by time $S_n^2$ on $M_1$ and $M_2$.
Furthermore, due to $p_{i_1} \le p_1 < q_n$, all the jobs of $\mcF^1$ can be processed in the same LPT order on $M_3$ and finished by time $C_n^2$.
Also, if the jobs of $\mcF^2$ are processed in the LPT order on $M_1$ following the jobs of $\mcF^1$,
then by time $C_n^2$, the job $J_{i_2}$ will be finished on $M_1$ and thus can start on $M_2$ at time $C_n^2$.
It follows that if the jobs of $\mcF^2$ are processed in the same LPT order on $M_2$ and $M_3$ in the schedule $\pi$,
then they are finished by time $C_n^2 + p_{i_2} + P(\mcF^2)$ on $M_2$ and $M_3$.
This way, the makespan of such a schedule $\pi$ is $C_{\max}^{\pi} = \max\{p_{i_1} + P(\mcF^1), q_n\} + q_n + \max\{p_{i_2} + P(\mcF^2), q_n\}$.
Note that the schedule can be constructed in $O(n \log n)$ time.
This finishes the proof of the lemma.
\end{proof}

\subsection{The $(1 + \epsilon)$-approximation algorithm $C(\epsilon)$}
%--------------------------------------------------------------------------------------------------
From each job $J_i \in \mcF$ with $i > 1$ and a bipartition $\{\mcC, \overline{\mcC}\}$ of the job subset $\{J_{i+1}, J_{i+2}, \ldots, J_{n-1}\}$,
we obtain a bipartition $\{\mcF^1, \mcF^2\}$ for $\mcF$ where $\mcF^1 = \{J_i\} \cup \mcC$ and $\mcF^2 = \{J_1, J_2, \ldots, J_{i-1}\} \cup \overline{\mcC}$.
Note that $J_i$ and $J_1$ are the largest job of $\mcF^1$ and $\mcF^2$, respectively.
We may then use Lemma~\ref{lemma64} to construct in $O(n \log n)$ time an optimal schedule $\pi^*$ in which 
the machine order for $J_n$ is forced to be $\langle M_3, M_2, M_1\rangle$,
and $\mcF^1, \mcF^2 \subset \mcF$ are the subsets of jobs forced to be processed before and after $J_n$ on the machine $M_2$, respectively.
The makespan of $\pi^*$ is $C_{\max}^* = \max\{p_i + P(\mcF^1), q_n\} + q_n + \max\{p_1 + P(\mcF^2), q_n\}$ (see for an illustration in Figure~\ref{fig64}).
Apparently the schedule $\pi^*$ relies on the bipartition $\{\mcF^1, \mcF^2\}$ of $\mcF$,
and it eventually relies on the job $J_i$ and the bipartition $\{\mcC, \overline{\mcC}\}$ of $\{J_{i+1}, J_{i+2}, \ldots, J_{n-1}\}$.
We therefore re-denote this schedule as $\pi^i(\mcC, \overline{\mcC})$.

Let $\pi^i$ denote the best schedule among all $\pi^i(\mcC, \overline{\mcC})$'s,
over all possible bipartitions $\{\mcC, \overline{\mcC}\}$ of $\{J_{i+1}, J_{i+2}, \ldots, J_{n-1}\}$.
Correspondingly, its makespan is denoted as $C_{\max}^i$.

When $2 p_i \ge q_n$, we have $2 p_1 \ge q_n$ too;
consequently for any bipartition $\{\mcC, \overline{\mcC}\}$, the makespan of the schedule $\pi^i(\mcC, \overline{\mcC})$ is
$p_i + P(\mcF^1) + q_n + p_1 + P(\mcF^2) = p_i + p_1 + q_n + P(\mcF)$.
We thus have proved the following lemma.

\begin{lemma}
\label{lemma65}
When $2 p_i \ge q_n$, the best schedule $\pi^i$ can be constructed in $O(n \log n)$ time and its makespan is $C_{\max}^i = p_i + p_1 + q_n + P(\mcF)$.
\end{lemma}

When $2 p_i < q_n$, the best schedule $\pi^i$ or the best bipartition $\{\mcC, \overline{\mcC}\}$ of $\{J_{i+1}, J_{i+2}, \ldots, J_{n-1}\}$ is hard to locate
(see Theorem~\ref{thm51}).
Nonetheless, we are able to design a fully polynomial time algorithm $B(\epsilon)$ that constructs a feasible schedule $\pi^{i, \epsilon}$
with its makespan $C_{\max}^{i, \epsilon} \le (1 + \epsilon) C_{\max}^i$, for any $\epsilon > 0$.

The algorithm $B(\epsilon)$ first constructs an instance of the {\sc Knapsack} problem,
in which an item $J_j$ ($j = i+1, i+2, \ldots, n-1$) corresponds to the job $J_j$, has a profit $p_j$ and has a size $p_j$,
and the capacity of the knapsack is set to $q_n - 2 p_i$.
It then calls an $O(n \min\{\log n, \log{(1/\epsilon)}\} + 1/{\epsilon}^2 \log{(1/\epsilon)} \min\{n, 1/{\epsilon} \log{(1/\epsilon)}\})$-time  
$(1 + \epsilon)$-approximation for the {\sc Min-Knapsack} problem to obtain a job subset $\mcC$;
and calls an $O(n \min\{\log n, \log{(1/\epsilon)}\} + 1/{\epsilon}^2 \log{(1/\epsilon)} \min\{n, 1/{\epsilon} \log{(1/\epsilon)}\})$-time  
$(1 - \epsilon)$-approximation for the {\sc Max-Knapsack} problem to obtain another job subset $\mcD$.
In another $O(n \log n)$ time, the algorithm constructs two schedules $\pi^i(\mcC, \overline{\mcC})$ and $\pi^i(\mcD, \overline{\mcD})$,
and returns the better one.
A high-level description of the algorithm $B(\epsilon)$ is provided in Figure~\ref{fig65}.

\begin{figure}[ht]
\begin{center}
\framebox{
\begin{minipage}{5.0in}
	{\sc Algorithm}  $B(\epsilon)$:
	\begin{enumerate}%[{Step} 1.]
		\item Construct an instance of {\sc Knapsack}, 
		      where an item $J_j$ corresponds to the job $J_j \in \{J_{i+1}, J_{i+2}, \ldots, J_{n-1}\}$;
		      $J_j$ has a profit $p_j$ and a size $p_j$;
		      the capacity of the knapsack is $q_n - 2p_i$. 
			\begin{description}
			\parskip=0pt
		    \item[1.1.] Run a $(1+\epsilon)$-approximation for {\sc Min-Knapsack} to obtain a job subset $\mcC$. 
		    \item[1.2.]	Use Lemma~\ref{lemma64} to construct a schedule $\pi^i(\mcC, \overline{\mcC})$.
			\end{description}
        \item 
			\begin{description}
			\parskip=0pt
		    \item[2.1.] Run a $(1-\epsilon)$-approximation for {\sc Max-Knapsack} to obtain a job subset $\mcD$.
		    \item[2.2.]	Use Lemma~\ref{lemma64} to construct a schedule $\pi^i(\mcD, \overline{\mcD})$.
			\end{description}
        \item Output the better schedule between $\pi^i(\mcC, \overline{\mcC})$ and $\pi^i(\mcD, \overline{\mcD})$.
	\end{enumerate}
\end{minipage}}
\end{center}
\caption{A high-level description of the algorithm $B(\epsilon)$.\label{fig65}}
\end{figure}

\begin{lemma}
\label{lemma66}
When $2 p_i < q_n$, the algorithm $B(\epsilon)$ constructs a feasible schedule $\pi^{i, \epsilon}$ in 
$O(n \min\{\log n$, $\log{(1/\epsilon)}\} + 1/{\epsilon}^2 \log{(1/\epsilon)} \min\{n, 1/{\epsilon} \log{(1/\epsilon)}\})$ time,
with its makespan $C_{\max}^{i, \epsilon} \le (1 + \epsilon) C_{\max}^i$, for any $\epsilon > 0$.
\end{lemma}
\begin{proof}
The running time of the algorithm $B(\epsilon)$ is dominated by the two calls to the approximation algorithms for the {\sc Knapsack} problem,
which are in $O(n \min\{\log n, \log{(1/\epsilon)}\} + 1/{\epsilon}^2 \log{(1/\epsilon)}$ $\min\{n, 1/{\epsilon} \log{(1/\epsilon)}\})$.

Let $\OPT^3$ ($\OPT^4$, respectively) denote the optimal solution to the {\sc Min-Knapsack} ({\sc Max-Knapsack}, respectively) problem
and also abuse it to denote the total profit of the items in the solution.
We therefore have the following (in-)equalities inside the algorithm $B(\epsilon)$:
\begin{eqnarray}
\OPT^3 		&= &\min \{ P(\mcX) \mid \mcX \subseteq \{J_{i+1}, J_{i+2}, \ldots, J_{n-1}\},~ P(\mcX) > q_n - 2p_i \};\label{eq8}\\
q_n - 2p_i  &< &P(\mcC) \ \ \le \ \ (1+\epsilon) \OPT^3;\label{eq9}\\
\OPT^4 		&= &\max \{ P(\mcY) \mid \mcY \subseteq \{J_{i+1}, J_{i+2}, \ldots, J_{n-1}\},~ P(\mcY) \le q_n - 2p_i \};\label{eq10}\\
q_n - 2p_i  &\ge &P(\mcD) \ \ \ge \ \ (1-\epsilon) \OPT^4.\label{eq11}
\end{eqnarray}

Let $C_{\max}^1$ ($C_{\max}^2$, respectively) denote the makespan of the schedule $\pi^i(\mcC, \overline{\mcC})$ ($\pi^i(\mcD, \overline{\mcD})$, respectively).
That is,
\begin{eqnarray}
C_{\max}^1 &= &2p_i + P(\mcC) + q_n + \max\left\{p_1 + \sum_{j=1}^{i-1} p_j + P(\overline{\mcC}), q_n\right\};\label{eq12}\\
C_{\max}^2 &= &2q_n + \max\left\{p_1 + \sum_{j=1}^{i-1} p_j + P(\overline{\mcD}), q_n\right\}.\label{eq13}
\end{eqnarray}
We distinguish three cases:

In the first case where $p_1 + \sum_{j=1}^{i-1} p_j + P(\overline{\mcC}) \ge q_n$, Eq.~(\ref{eq12}) becomes
\[
C_{\max}^1 = 2 p_i + P(\mcC) + q_n + p_1 + \sum_{j=1}^{i-1} p_j + P(\overline{\mcC}) = p_i + p_1 + P(\mcF) + q_n \le C_{\max}^i,
\]
where the last inequality holds by Lemma~\ref{lemma64}, suggesting that the schedule $\pi^i(\mcC, \overline{\mcC})$ is optimal.

In the second case where $p_1 + \sum_{j=1}^{i-1} p_j + P(\overline{\mcD}) \le q_n$, Eq.~(\ref{eq13}) becomes
\[
C_{\max}^2 = 3q_n \le C_{\max}^i,
\]
where the last inequality holds by Lemma~\ref{lemma64} again, suggesting that the schedule $\pi^i(\mcD, \overline{\mcD})$ is optimal.

In the last case, we consider $p_1 + \sum_{j=1}^{i-1} p_j + P(\overline{\mcC}) < q_n$ and $p_1 + \sum_{j=1}^{i-1} p_j + P(\overline{\mcD}) > q_n$.
Assume in the best schedule $\pi^i$ the bipartition of $\{J_{i+1}, J_{i+2}, \ldots, J_{n-1}\}$ is $\{\mcC^*, \overline{\mcC^*}\}$.

If $P(\mcC^*) > q_n - 2p_i$, {\it i.e.}, $\mcC^*$ is a feasible solution to the constructed {\sc Min-Knapsack} instance, then we have
\begin{equation}
\label{eq14}
\OPT^3 \le P(\mcC^*).
\end{equation}
Using Eqs.~(\ref{eq9}, \ref{eq14}), Eq.~(\ref{eq12}) becomes
\begin{equation}
\label{eq15}
\begin{array}{lll}
C_{\max}^1 	&= \ 2q_n + 2p_i + P(\mcC)\\
			&\le \ 2q_n + 2p_i + (1 + \epsilon) \OPT^3\\
        	&\le \ 2q_n + 2p_i + (1 + \epsilon) P(\mcC^*)\\
		 	&\le \ (1 + \epsilon) (2q_n + 2p_i + P(\mcC^*))\\
        	&\le \ (1 + \epsilon) C_{\max}^i,
\end{array} 
\end{equation}
where the last inequality holds by Lemma~\ref{lemma64} again.

Otherwise we have $P(\mcC^*) \le q_n - 2p_i$, {\it i.e.}, $\mcC^*$ is a feasible solution to the constructed {\sc Max-Knapsack} instance;
we have
\begin{equation}
\label{eq16}
\OPT^4 \ge P(\mcC^*).
\end{equation}
Using Eqs.~(\ref{eq11}, \ref{eq16}), Eq.~(\ref{eq13}) becomes
\begin{equation}
\label{eq17}
\begin{array}{ll}
C_{\max}^2 	&= \ 2q_n + p_1 + \sum_{j=1}^{i-1} p_j + P(\overline{\mcD})\\
			&= \ 2q_n + p_1 + \sum_{j=1}^{i-1} p_j + P(\overline{\mcC^*}) + P(\overline{\mcD}) - P(\overline{\mcC^*})\\
			&= \ 2q_n + p_1 + \sum_{j=1}^{i-1} p_j + P(\overline{\mcC^*}) + P(\mcC^*) - P(\mcD)\\
			&\le \ 2q_n + p_1 + \sum_{j=1}^{i-1} p_j + P(\overline{\mcC^*}) + \epsilon P(\mcC^*)\\
			&\le \ C_{\max}^i + \epsilon P(\mcC^*)\\
        	&\le \ (1 + \epsilon) C_{\max}^i,
\end{array} 
\end{equation}
where the second last inequality holds by Lemma~\ref{lemma64} and the last inequality holds due to the trivial fact that $P(\mcC^*) \le C_{\max}^i$.

Therefore, in the last case, combining Eqs.~(\ref{eq15}, \ref{eq17}) we have
\[
C_{\max}^{i,\epsilon} = \min\{C_{\max}^1, C_{\max}^2\} \le (1 + \epsilon) C_{\max}^i.
\]
This finishes the proof of the lemma.
\end{proof}

The $(1 + \epsilon)$-approximation algorithm $C(\epsilon)$ for the $M3 \mid prpt, (n-1, 1) \mid C_{\max}$ problem takes advantage of
the $(1 + \epsilon)$-approximation algorithm $A(\epsilon)$ for the $M3 \mid prpt \mid C_{\max}$ problem when $p_1 \ge q_{\ell+1}$
(presented in Section 3, see Theorem~\ref{thm37}),
the optimal schedule constructed in Lemma~\ref{lemma61} for the $M3 \mid prpt, (n-1, 1), p_1 < q_n \mid C_{\max}$ problem when $P(\mcF) \le q_n$,
the two optimal schedules constructed in Lemma~\ref{lemma63} for the $M3 \mid prpt, (n-1, 1), p_1 < q_n \mid C_{\max}$ problem
when the machine order for $J_n$ is forced to be $\langle M_2, M_3, M_1\rangle$ or $\langle M_3, M_1, M_2\rangle$, respectively,
denoted as $\pi^{2,3,1}$ and $\pi^{3,1,2}$, respectively,
and the $n-2$ schedules $\pi^i$, $i = 2, 3, \ldots, n-1$, constructed either in Lemma~\ref{lemma65} or by the algorithm $B(\epsilon)$
(see Lemma~\ref{lemma66}).
A high-level description of $C(\epsilon)$ is depicted in Figure~\ref{fig66}.

\begin{figure}[ht]
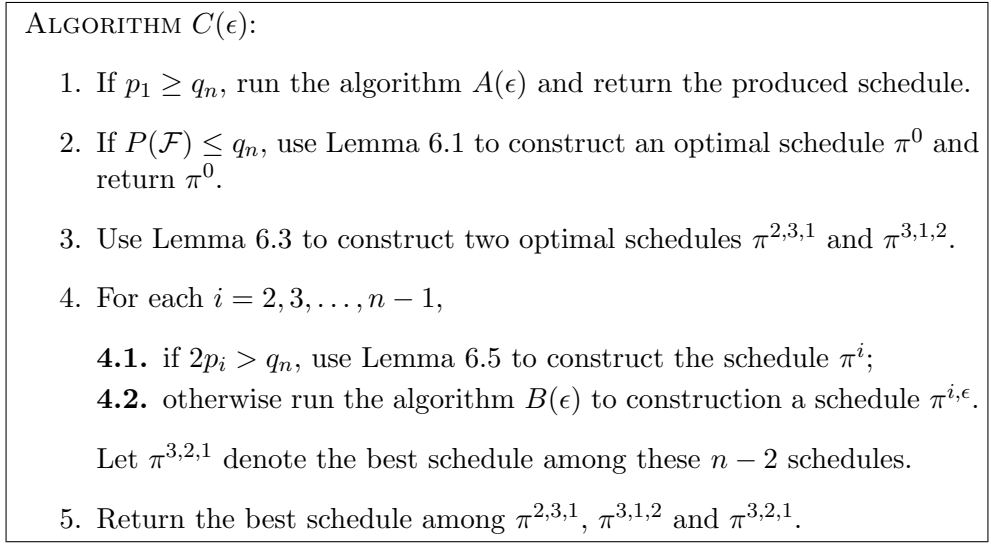

\begin{center}
\framebox{
\begin{minipage}{5.0in}
{\sc Algorithm}  $C(\epsilon)$:
\begin{enumerate}%[{Step} 1.]
\item
	If $p_1 \ge q_n$, run the algorithm $A(\epsilon)$ and return the produced schedule.
\item
	If $P(\mcF) \le q_n$, use Lemma~\ref{lemma61} to construct an optimal schedule $\pi^0$ and return $\pi^0$.
\item
	Use Lemma~\ref{lemma63} to construct two optimal schedules $\pi^{2,3,1}$ and $\pi^{3,1,2}$.
\item
	For each $i = 2, 3, \ldots, n-1$,
	\begin{description}
	\parskip=0pt
    \item[4.1.] if $2p_i > q_n$, use Lemma~\ref{lemma65} to construct the schedule $\pi^i$;
    \item[4.2.]	otherwise run the algorithm $B(\epsilon)$ to construction a schedule $\pi^{i,\epsilon}$.
	\end{description}
	Let $\pi^{3,2,1}$ denote the best schedule among these $n-2$ schedules.
\item
	Return the best schedule among $\pi^{2,3,1}$, $\pi^{3,1,2}$ and $\pi^{3,2,1}$.
\end{enumerate}
\end{minipage}}
\end{center}
\caption{A high-level description of the algorithm $C(\epsilon)$.\label{fig66}}
\end{figure}

\begin{theorem}
\label{thm67}
The algorithm $C(\epsilon)$ is an $O(n^2 \min\{\log n, \log{(1/\epsilon)}\} + n/{\epsilon}^2 \log{(1/\epsilon)} \min\{n, 1/{\epsilon} \log{(1/\epsilon)}\})$-time
$(1 + \epsilon)$-approximation for the $M3 \mid prpt, (n-1, 1) \mid C_{\max}$ problem.
\end{theorem}
\begin{proof}
Recall from Theorem~\ref{thm37} that the running time of the algorithm $A(\epsilon)$ is in
$O(n \min\{\log n$, $\log{(1/\epsilon)}\} + 1/{\epsilon}^2 \log{(1/\epsilon)} \min\{n, 1/{\epsilon} \log{(1/\epsilon)}\})$;
Lemma~\ref{lemma61} construct the optimal schedule $\pi^0$ in $O(n)$ time;
Lemma~\ref{lemma63} constructs each of the two optimal schedules $\pi^{2,3,1}$ and $\pi^{3,1,2}$ in $O(n \log n)$ time;
Lemma~\ref{lemma65} construct the schedule $\pi^i$ also in $O(n \log n)$ time;
and the running time of the algorithm $B(\epsilon)$ is in
$O(n \min\{\log n, \log{(1/\epsilon)}\} + 1/{\epsilon}^2 \log{(1/\epsilon)} \min\{n, 1/{\epsilon} \log{(1/\epsilon)}\})$.
Since $A(\epsilon)$ is called only once, while $B(\epsilon)$ could be called for $O(n)$ times,
the overall time complexity of $C(\epsilon)$ is an order higher, that is,
$O(n^2 \min\{\log n, \log{(1/\epsilon)}\} + n/{\epsilon}^2 \log{(1/\epsilon)} \min\{n, 1/{\epsilon} \log{(1/\epsilon)}\})$.

The worst-case performance ratio of $1 + \epsilon$ is implied by Theorem~\ref{thm37} and Lemmas~\ref{lemma61}--\ref{lemma66}.
In more details, if $p_1 \ge q_n$, then the ratio is guaranteed by Theorem~\ref{thm37};
if $P(\mcF) \le q_n$, then the optimality is guaranteed by Lemma~\ref{lemma61};
otherwise, Lemma~\ref{lemma62} states that in the optimal schedule, the machine order for the unique open-shop job $J_n$ is
one of $\langle M_2, M_3, M_1\rangle$, $\langle M_3, M_1, M_2\rangle$, and $\langle M_3, M_2, M_1\rangle$:
in the first two cases, the optimality is guaranteed by Lemma~\ref{lemma63},
while in the last case, the performance ratio is guaranteed by Lemmas~\ref{lemma65} and \ref{lemma66} together.
This finishes the proof of the theorem.
\end{proof}

\section{Concluding remarks} \label{sec7}
%==================================================================================================
In this paper, we studied the three-machine proportionate mixed shop problem $M3 \mid prpt \mid C_{\max}$.
We presented first an FPTAS for the case where $p_1 \ge q_{\ell + 1}$;
and then proposed a $4/3$-approximation algorithm for the other case where $p_1 < q_{\ell + 1}$,
for which we also showed that the performance ratio of $4/3$ is asymptotically tight.
The $F3 \mid prpt \mid C_{\max}$ problem is polynomial-time solvable;
we showed an interesting hardness result that adding only one open-shop job to the job set makes the problem NP-hard
if the open-shop job is larger than any flow-shop job.
The special case in which there is only one open-shop job is denoted as $M3 \mid prpt, (n-1, 1) \mid C_{\max}$.
Lastly we proposed an FPTAS for this special case $M3 \mid prpt, (n-1, 1) \mid C_{\max}$.

We believe that when $p_1 < q_{\ell + 1}$, the $M3 \mid prpt \mid C_{\max}$ problem can be better approximated than $4/3$,
and an FPTAS is perhaps possible.
Our last FPTAS for the special case $M3 \mid prpt, (n-1, 1) \mid C_{\max}$ can be considered as the first successful step towards such an FPTAS.

\subsection*{Acknowledgement}
%==============================================================================
%The authors would like to thank the anonymous reviewers for their many suggestions and comments that help improve the paper presentation.
LL was supported by the China Scholarship Council Grant 201706315073 and the Fundamental Research Funds for the Central Universities Grant No. 20720160035.
YC and AZ were supported by the NSFC Grants 11771114 and 11571252;
YC was also supported by the China Scholarship Council Grant 201508330054.
JD was supported by the NSFC Grant 11501512 and the Zhejiang Provincial Natural Science Foundation Grant No. LY18A010029.
RG and GL were supported by the NSERC Canada;
LL, YC, JD, GN and AZ were all partially supported by the NSERC Canada during their visits to Edmonton.
GN was supported by the NSFC Grant 71501045, the Natural Science Foundation of Fujian Province Grant 2016J01332 and the Education Department of Fujian Province.

%\bibliography{BiBTeX/mybibfile,BiBTeX/scheduling,BiBTeX/mypapers,BiBTeX/general,BiBTeX/ppp}

\end{document}